\documentclass[a4paper,UKenglish,cleveref, autoref, thm-restate]{lipics-v2021}
\pdfoutput=1
\usepackage{amsmath,amssymb}
\usepackage{amsthm}
\usepackage[utf8]{inputenc}
\usepackage{todonotes}

\usepackage{thmtools}
\usepackage{chngcntr}

\usepackage{tikz}
\usepackage{float}
\usepackage{cancel}
\usepackage{algorithm}
\usepackage{algpseudocode}

\usetikzlibrary{arrows,automata}

\def\N{\mathbb{N}}

\def\ta{\mathtt{a}}
\def\tb{\mathtt{b}}
\def\tc{\mathtt{c}}

\def\tx{\mathtt{x}}

\DeclareMathOperator{\Subseq}{SubSeq}
\DeclareMathOperator{\Univ}{Univ}
\DeclareMathOperator{\PerfUniv}{PUniv}
\DeclareMathOperator{\PUniv}{PUniv}

\DeclareMathOperator{\letters}{alph}

\DeclareMathOperator{\ar}{ar}

\def\r{\operatorname{r}}

\DeclareMathOperator{\al}{alph}

\def\nth#1{#1$^{\text{th}}$}

\theoremstyle{definition}
\newtheorem{problem}[definition]{Problem}

\newif\ifpaper
\paperfalse 
\nolinenumbers

\title{$k$-Universality of Regular Languages} 

\titlerunning{$k$-Universality of Regular Languages} 

\author{Duncan Adamson}{Leverhulme Centre for Functional Material Design, University of Liverpool, UK }{d.a.adamson@liverpool.ac.uk}{0000-0003-3343-2435}{}
\author{Pamela Fleischmann}{Department of Computer Science, Kiel University, 
Germany}{fpa@informatik.uni-kiel.de}{https://orcid.org/0000-0002-1531-7970}{}
\author{Annika Huch}{Department of Computer Science, Kiel University, Germany}{stu216885@mail.uni-kiel.de}{}{}
\author{Tore Koß}{Department of Computer Science, University of Göttingen, Germany}{tore.koss@cs.uni-goettingen.de}{https://orcid.org/0000-0001-6002-1581}{}
\author{Florin Manea}{Department of Computer Science, University of Göttingen, Germany}{florin.manea@cs.uni-goettingen.de}{https://orcid.org/0000-0001-6094-3324}{}
\author{Dirk Nowotka}{Department of Computer Science, Kiel University, 
Germany}{dn@informatik.uni-kiel.de}{https://orcid.org/0000-0002-5422-2229}{}

\authorrunning{D.~Adamson, P.~Fleischmann, A.~Huch, T.~Koß, F.~Manea and D.~Nowotka} 

\Copyright{Duncan Adamson, Pamela Fleischmann, Annika Huch, Tore Koß, Florin Manea and Dirk Nowotka} 



\category{} 

\relatedversion{} 

\supplement{}

\EventEditors{Satoru Iwata and Naonori Kakimura}
\EventNoEds{2}
\EventLongTitle{34th International Symposium on Algorithms and Computation (ISAAC 2023)}
\EventShortTitle{ISAAC 2023}
\EventAcronym{ISAAC}
\EventYear{2023}
\EventDate{December 3-6, 2023}
\EventLocation{Kyoto, Japan}
\EventLogo{}
\SeriesVolume{283}
\ArticleNo{3}

\begin{document}

\hideLIPIcs

\maketitle

\begin{abstract}
A subsequence of a word $w$ is a word $u$ such that $u = w[i_1] w[i_2]  \dots w[i_{k}]$, for some set of indices 
$1 \leq i_1 < i_2 < \dots < i_k \leq \lvert w\rvert$.
A word $w$ is $k$-subsequence universal over an alphabet $\Sigma$ if every word in $\Sigma^k$ appears in $w$ as a subsequence.
In this paper, we study the intersection between the set of $k$-subsequence universal words over some alphabet $\Sigma$ and
regular languages over $\Sigma$.
We call a regular language $L$ \emph{$k$-$\exists$-subsequence universal} if there exists a $k$-subsequence universal word in
$L$, and \emph{$k$-$\forall$-subsequence universal} if every word of $L$ is $k$-subsequence universal.
We give algorithms solving the problems of deciding if a given regular language, represented by a finite automaton recognising 
it, is \emph{$k$-$\exists$-subsequence universal} and, respectively, if it is \emph{$k$-$\forall$-subsequence universal}, for a given $k$. 
The algorithms are FPT w.r.t.~the size of the input alphabet, and their run-time does not depend on $k$; they run in polynomial time in the number $n$ of states of the input automaton when the size of the input alphabet is $O(\log n)$. Moreover, we show that the problem of deciding if a given regular language is \emph{$k$-$\exists$-subsequence universal} is NP-complete, when the language is over a large alphabet. Further, we provide algorithms for counting the number of $k$-subsequence universal words (paths) accepted by a given deterministic (respectively, nondeterministic) finite automaton, and ranking an input word (path) within the set of $k$-subsequence universal words accepted by a given finite automaton. 
\looseness=-1

\end{abstract}

\section{Introduction}
Words and subsequences are two fundamental combinatorial objects.
Informally, a subsequence of a word $w$ is a word $u$ that can be obtained by deleting some of $w$'s letters while preserving the order of the rest. For instance, $\mathtt{taunt}$ and $\mathtt{salty}$ are subsequences of $\mathtt{automatauniversality}$, while $\mathtt{trauma}$ is not since these letters do not occur in the correct order.
Subsequences are a heavily studied object within computer science \cite{barker2020scattered,day2021edit,fleischmann2022nearly,halfon2017decidability,KimKH22,kosche2021absent,lothaire,mateescu2004subword,SchnoebelenV23,simon2003words,tronicek2003common,zetzsche2016complexity} and beyond, with applications in a wide number of fields including bioinformatics \cite{han2020novel,shikder2019openmp}, database theory \cite{artikis2017complex,FrochauxK23,Kleest-Meissner22,Kleest-Meissner23}, and modelling concurrency \cite{shaw1978software}.
A survey of combinatorial pattern matching algorithms for subsequences has been provided by Kosche et al. \cite{Kosche2022SubsequenceSurvey}, highlighting a series of recent results for problems on finding subsequences in words as well as their applications and connections to other areas of computer science, to which we refer the reader for further details and references.

This paper considers \emph{$k$-subsequence universal} words.
A word $w$ is $k$-subsequence universal over an alphabet $\Sigma=\{1,\ldots,\sigma\}$ if $w$ contains every word of length $k$ over $\Sigma$ as a subsequence. 
These words were first defined by Karandikar and Schnoebelen \cite{karandikar2016height,schnoebelen2019height} as \emph{$k$-rich words}, however more recent work has used the term \emph{$k$-subsequence} (or {\em scattered factor}) {\em universality} \cite{adamson2023words,Goettingen2023words,barker2020scattered,day2021edit,fleischmann2021scattered,fleischmann2023alphabetafactorization,kosche2021absent,SchnoebelenV23}, which we use here. The study of these words follows from the seminal work by Simon \cite{Simon72} where a congruence - nowadays known as {\em Simon's congruence} - is introduced.
Two words $w,v$ are $k$-congruent, denoted $w\sim_k v$, if $w$ and $v$ share the same set of subsequences up to length $k$. As such, $k$-subsequence universal words are those which are $k$-congruent to the word $(1\cdots \sigma)^k$. 

Simon's congruence relation is well studied \cite{fleischer2018testing,fleischmann2022nearly,fleischmann2023alphabetafactorization,simon2003words,tronicek2003common,zetzsche2016complexity}, with recent asymptotically optimal algorithms for testing if two words are $k$-congruent \cite{barker2020scattered} and for computing the largest $k$ for which two words are $k$-congruent \cite{gawrychowski2021simons}, as well as for pattern matching under Simon's congruence \cite{KimKH22}. 
Indeed, besides the usage of $k$-subsequence universal words in~\cite{karandikar2016height,schnoebelen2019height} in a context related to the study of the height of piecewise testable languages and the logic of subsequence, the idea of universality itself is quite important in formal languages, automata theory, but also in combinatorics. In this context, the universality problem \cite{HolzerK11} is whether a given language $L$  (over an alphabet $\Sigma$, given as an automaton) is equal to $\Sigma^{\ast}$. This problem for various classes of languages and language accepting/generating formalisms is studied in, e.g., \cite{Rampersad:2012,KrotzschMT17,GawrychowskiRSS17} and the references therein. The universality problem was considered for contiguous factors (substrings) of words \cite{martin1934,Bruijn46} and  partial words \cite{ChenKMS17,GoecknerGHKKKS18}, as well. In that context, one analyses, the (partial) words $w$ over an alphabet $\Sigma$ which have exactly one occurrence of each string of length $\ell$ over $\Sigma$ as a substring. De Bruijn sequences~\cite{Bruijn46} fulfil this property and have many applications in computer science or combinatorics, see \cite{ChenKMS17,GoecknerGHKKKS18} and the references therein. While it is a perfectly valid problem to investigate (partial) words where each substring occurs exactly once, in the case of subsequence universality this is a trivial restriction, as in each long-enough word there will be subsequences occurring more than once \cite{barker2020scattered}.\looseness=-1

In a similar direction, the work of works of Zetzsche \cite{zetzsche2016complexity}, Karandikar et al. \cite{karandikar2016state}, and Bachmeier et al. \cite{bachmeier2015finite} compare the \emph{downward closure} of languages. The downward closure of a word $w$ is the set of all subwords of $w$. By this definition, two words are $k$-congruent if the subset of the downward closures of each word containing those words of length at most $k$ are equivilent. In \cite{bachmeier2015finite,karandikar2016state,zetzsche2016complexity}, the downward closure of a language $L$ is defined as the set of all words appearing as a subsequence of at least one word in $L$. Most relevant to our work is that of Karandikar et al. \cite{karandikar2016state}, who show that the problem of deciding if the downward closure of some language $L$, defined by an automaton $A$, is a subset of another language $L'$, defined by an automaton $A'$, is CoNP-complete.
\looseness=-1

Coming closer to the topic of this work, we recall the works by Barker et al. \cite{barker2020scattered}, Day et al. \cite{day2021edit}, and Schnoebelen and Veron \cite{SchnoebelenV23}, which directly address $k$-subsequence universal words.
In \cite{barker2020scattered}, the authors show that it is possible to determine in linear time (1) whether a word is $k$-subsequence universal and (2) the shortest $k$-subsequence universal prefix of a given word. Additionally, they show that the minimal set of factors $w_1, w_2, \dots, w_{\ell}$ of a word $w$,  such that $w_1 w_2 \cdots w_{\ell}$ is $k$-subsequence universal and the exponent $i$ such that $w^i$ is $k$-subsequence universal can be determined efficiently.
Further, \cite{day2021edit} proposes a set of algorithmic results for computing the minimum number of edit operations (insertion, deletions, substitutions) to apply to a word $w$ in order to make it $k$-subsequence universal; in general, their algorithms run in $O(\vert w \vert k)$ time, for $k\leq |w|$ (the problems are trivial, otherwise). Interestingly, the problems approached in that paper can be seen as determining the minimum number $\Delta$ such that the (finite) language containing the words found at edit distance at most $\Delta$ from $w$ contains a $k$-universal word. Finally, \cite{SchnoebelenV23} presents an algorithm computing the largest $k$ for which a word, given in a compressed form, is $k$-universal. \looseness=-1

Another paper which is also strongly related to our work is \cite{KimHKS22} (as well as its follow-up \cite{KimHKS23}). In this paper, the authors investigate the language of words which are $k$-congruent to $u$, for a given word $u$, called the $k$-closure of $u$. They show that this language is regular, and effectively construct a finite automaton recognising it (of size exponential in $\sigma$, the size of the input alphabet). Now, testing whether another given finite (or regular) language contains a word which is $k$-congruent to $u$ can be reduced to deciding the emptiness of the intersection between the $k$-closure of $u$ and the given language. 

Our work builds on \cite{day2021edit,KimHKS22}, as well as on the works whose main object is the downward closure of languages (see \cite{zetzsche2016complexity} and the references therein), and investigates the problem of efficiently detecting $k$-subsequence universal words within regular languages. In other words, we are interested in identifying the words of a given regular language $L$ which are $k$-subsequence universal, i.e., in the $k$-closure of the target word $(1\cdots \sigma)^k$. 
There is, however, a fundamental difference between the problems considered here and those of  \cite{KimHKS22,KimHKS23}. The input of the problems studied here is a regular language $L$ and a number $k$ (given in binary representation, similarly to \cite{day2021edit}), and we want to detect the words of $L$ which are $k$-universal, without having to explicitly write those words or the target word $(1\cdots \sigma)^k$ (whose length is at least $k\sigma$, so exponential in the size of the binary representation of $k$, our input). On the other hand, in the setting of \cite{KimHKS22,KimHKS23} we are explicitly given a target word $w$ for which we construct the finite automaton recognising its $k$-closure; applying this approach to our setting would lead to dealing explicitly with the target word $w=(1\cdots \sigma)^k$, so we would directly have an exponential blow-up both at this step and when the automaton is constructed. Hence, the problem discussed in this paper extends significantly the one of \cite{day2021edit} by looking at the intersection between the language of $k$-subsequence universal words and arbitrary regular languages, rather than a very particular class of finite languages. Moreover, it addresses a highly-relevant particular case of the theory presented in \cite{KimHKS22,KimHKS23}, by considering the class of $k$-subsequence universal words, with the interesting property that this case can be succinctly specified, and with the hope that more direct and efficient algorithms can be obtained in this setting, without having to go through the general framework. 

Going more into the details of our approach, we start our investigation by defining two notions for $k$-subsequence universality of regular languages. A regular language $L$ is {\em existence $k$-subsequence universal} ($k$-$\exists$-subsequence universal) if there exists at least one $k$-subsequence universal word in $L$, and it is {\em universal $k$-subsequence universal}  ($k$-$\forall$-subsequence universal) if every word of $L$ is $k$-subsequence universal. We assume that regular languages are given by finite automata which recognises them, so, canonically, we will call an automaton $k$-$\exists$-subsequence universal (respectively, $k$-$\forall$-subsequence universal) if the language accepted by it is $k$-$\exists$-subsequence universal (respectively, $k$-$\forall$-subsequence universal).

Alongside the categorisation problems (given an automaton, decide whether the language it recognises is $k$-$\exists$-subsequence universal or $k$-$\forall$-subsequence universal), we consider the problems of {\em counting} and {\em ranking} the number of $\ell$-length $k$-subsequence universal words accepted by a given finite automaton $\mathcal{A}$.
The counting problem asks for the total number of $\ell$-length $k$-subsequence universal words accepted by $\mathcal{A}$.
The ranking problem takes a word $w$ as input and asks for the number of $k$-subsequence universal words accepted by $\mathcal{A}$ that are lexicographically smaller than $w$.
Both of these problems have been heavily studied for other classes of words, including cyclic words \cite{adamson2022ranking,adamson2023words,Adamson2021,Fredricksen1978,gilbert1961symmetry,Kociumaka2014,Sawada2017} and Gray codes \cite{Fredricksen1978,Kociumaka2014,savage1997survey}.

\subparagraph{Our Contributions.} In Section~\ref{prelims} we introduce the novel notions of  $k$-$\exists$-subsequence universality or $k$-$\forall$-subsequence universality for regular languages and finite automata.

In Section \ref{npc}, we give algorithms solving the problems of deciding if the language accepted by a given finite automaton with $n$ states is \emph{$k$-$\exists$-subsequence universal} and, respectively, if it is \emph{$k$-$\forall$-subsequence universal}, for a given $k$. 
Afterwards, for a given $k$, we give a polynomial time algorithm solving the problems of deciding if the language accepted by a given finite automaton with $n$ states is  \emph{$k$-$\forall$-subsequence universal}, and  a fixed parameter tractable algorithm (with respect to the size $\sigma$ of the input alphabet) for deciding if it is \emph{$k$-$\exists$-subsequence universal}.
If we have additionally $\sigma\in O(\log n)$, they run in polynomial time, and their run-time does not depend at all on $k$. Note that one could easily devise solutions for both these problems using the framework of \cite{KimHKS22,KimHKS23}, but their complexity would have been exponential both in $\log k$ (the size of the representation of $k$ in our input) and in $\sigma$. Moreover, we show that, if no bound is placed on the size of the input alphabet, the problem of deciding if a given regular language is \emph{$k$-$\exists$-subsequence universal} is NP-complete. The NP-hardness of this problem follows from \cite{KimHKS22}; it is worth noting that, on the one hand, the problem is hard even if $k=1$, but also, on the other hand, that the hardness proof is indeed based on the fact that the input alphabet is large (equal to the number of states in the input automaton). Showing that this problem is in NP, as well as our algorithms, requires a series of combinatorial insights on the structure of $k$-universal words accepted by finite automata.   \looseness=-1

Further, in Section \ref{sec:counting}, building on the aforementioned understanding of the combinatorial properties of $k$-universal words accepted by finite automata, we provide algorithms for counting the number of $k$-subsequence universal words (respectively, paths) accepted by a given deterministic (respectively, non-deterministic) finite automaton, and ranking an input word within the set of $k$-subsequence universal words (paths) accepted by a given deterministic (respectively, non-deterministic) finite automaton. Again, this approach extends non-trivially the approach from \cite{adamson2023words}, where problems related to counting and ranking subsequence universal words (unrestricted by any regular membership constraint) were approached for the first time.\looseness=-1

\section{Preliminaries}\label{prelims}
Let  $\N = \{1,2,\ldots\}$ denote the natural numbers and set $\N_0 = \N
\cup \{0\}$ as well as $[n]=\{1,\ldots,n\}$ and $[i,n]=\{i, i+1, \ldots, n\}$ for all $i,n\in\N_0$ with $i \leq n$.

An \emph{alphabet} $\Sigma=\{1,2,\ldots,\sigma\}$ is a finite set of symbols, called
\emph{letters} (w.l.o.g., we can assume that the letters are integers). A \emph{word} $w$ is a finite sequence of letters from a given alphabet and its length $|w|$ is the number of $w$'s letters. For $i \in
[|w|]$ let $w[i]$ denote $w$'s \nth{$i$} letter.  The set of all
finite words (aka strings) over the alphabet $\Sigma$, denoted by $\Sigma^{\ast}$, is the free monoid generated by $\Sigma$ with
concatenation as operation and the neutral element is the empty word
$\varepsilon$, i.e., the word of length $0$. Let $\Sigma^n$ denote all words in $\Sigma^{\ast}$ exactly of length $n\in\N_0$ and $\Sigma^{\leq n}$ the set of all words of $\Sigma^{\ast}$ up to length $n\in\N_0$.
Set $\letters(w) = \{\ta \in \Sigma \mid \exists i \in [|w|]: w[i] = \ta \}$ as $w$'s alphabet. 
For $u,w\in\Sigma^{\ast}$, $u$ is called a \emph{factor}
of $w$, if $w = xuy$ for some words $x,y\in\Sigma^{\ast}$. If $x = \varepsilon$ (resp.,
$y = \varepsilon$) then $u$ is called a \emph{prefix} (resp., \emph{suffix}) of
$w$.  For $1\leq i\leq j\leq|w|$ 
define the factor from $w$'s \nth{$i$} letter to the \nth{$j$} letter by  $w[i,j]=w[i]\cdots w[j]$.
Futher, given a pair of indices $i  < j$, $w[j, i] = \varepsilon$.
Let $<$ be an order relation on $\Sigma$ (e.g., the natural order on integers). We extend this order relation to the lexicrographical order on $\Sigma^{\ast}$ in the following way: the word $u$ is lexicographically smaller than the word $w$ ($u<w$) iff either $u$ is a prefix of $w$ or there exists $x,y_1,y_2\in\Sigma^{\ast}$ and $\ta,\tb\in\Sigma$ with $u=x\ta y_1$, $w=x\tb y_2$ and $\ta<\tb$.\looseness=-1



As we are interested in investigating the $k$-subsequence universality of regular languages, we firstly introduce the basic concepts related to subsequences. We then present the definitions for the transformation of these notions to the domain of regular languages and finite automata.

\begin{definition}
	Let $w \in \Sigma^*$ and $n \in \N_0$. A word $u \in \Sigma^*$ is called \emph{subsequence} of $w$ ($u \in \Subseq(w)$) if there exist $v_1, \ldots, v_{n+1} \in \Sigma^*$ such that $w = v_1 u[1] v_2 u[2] \cdots v_n u[n] v_{n+1}$. Set $\Subseq_k(w) = \{u \in \Subseq(w) \mid \vert u \vert = k\}$.
\end{definition}

\begin{example}
	Subsequences of $\mathtt{automatauniversality}$ are $\mathtt{auto}$, $\mathtt{tomata}$, $\mathtt{salty}$,  $\mathtt{mate}$, and $\mathtt{atom}$ while $\mathtt{star}$ and $\mathtt{alien}$ are not because their letters do not occur in the correct order.
\end{example}

In \cite{barker2020scattered}, the authors investigated words which have, for a given $k\in\N_0$, all words from $\Sigma^{k}$ as subsequence, namely $k$-subsequence universal words. Note that this notion is similar to the one of richness introduced and investigated in \cite{KarandikarKS15,karandikar2016height}. We stick here to the notion of $k$-subsequence universality since our focus are regular languages and thus the well-known notion of the universality of automata and formal languages, i.e., $L(\mathcal{A})=\Sigma^{\ast}$ for a given finite automaton $\mathcal{A}$, is close to the one of subsequence universality of words.

\begin{definition}\label{def:subsequence}
	A word $w \in \Sigma^*$ is called \emph{$k$-subsequence universal} (w.r.t. $\Sigma$), for $k \in \N_0$, if	
	$\Subseq_k(w) = \Sigma^k$. If the context is clear we briefly call $w$ $k$-universal. The universality-index $\iota(w)$  is the largest $k$ such that $w$ is $k$-universal.
	
\end{definition}

We denote the set of $k$-universal words in a given set $\mathcal{M}\subseteq\Sigma^{\ast}$ by $\Univ_{\mathcal{M}, k}$. Thus, the set of all $k$-universal words over a given alphabet $\Sigma$ is denoted $\Univ_{\Sigma^*, k}$.

\begin{example}\label{prelims:exampleuniversality}
	Consider the word $w=\mathtt{baaababb} \in \{\ta,\tb\}^*$. We have $\vert \Subseq_3 (\mathtt{baaababb})\vert = \vert \{\mathtt{aaa}, \mathtt{aab}, \mathtt{aba}, \mathtt{abb}, \mathtt{baa}, \mathtt{bab}, \mathtt{bba}, \mathtt{bbb}\} \vert = 8 = 2^3$. Thus, $\mathtt{baaababb} \in \Univ_{\{\ta,\tb\}^*, 3}$. Since $\mathtt{abba} \notin \Subseq_4(\mathtt{baaababb})$, it follows that $\mathtt{baaababb} \notin \Univ_{\{\ta,\tb\}^*, 4}$ and $\iota(\mathtt{baaababb}) = 3$.
\end{example}

Further, we recall the {\em arch factorisation} by Hébrard \cite{hebrard1991algorithm} which factorises words uniquely.

\begin{definition}
The {\em arch factorisation} of $w \in \Sigma^*$ is given by $w = \ar_1(w) \cdots \ar_k(w) \r(w)$ for $k \in \mathbb{N}_0$ with
	$\iota(\ar_i(w))=1$ and $\ar_i(w)[\vert \ar_i(w) \vert] \notin \letters(\ar_i(w)[1, \vert \ar_i(w) \vert - 1 ])$ for all $i \in [k]$, as well as $\letters(\r(w)) \subsetneq \Sigma$.
	The words $\ar_i(w)$ are the \emph{arches} and $\r(w)$ is the \emph{rest} of $w$.
\end{definition}

\begin{example}\label{prelims:exampleuniversalitycondt}
	Continuing Example~\ref{prelims:exampleuniversality}, we have the arch factorisation $w=\mathtt{(ba) \cdot (aab) \cdot (ab) \cdot b}$ where the parantheses denote the three arches and the rest $\mathtt{b}$.
\end{example}

A proper subset of the $k$-universal words is given by all the words with an empty rest
introduced in \cite{fleischmann2021scattered} as the set of {\em perfect $k$-universal} words. 

\begin{definition}
	We call a word $w \in \Sigma^*$ \emph{perfect $k$-universal} if $\iota(w)=k$ and $\r(w) = \varepsilon$. The set of all these words with $\letters(w)=\Sigma$ is denoted by $\PerfUniv_{\Sigma^\ast,k}$.
\end{definition}

\begin{example}
	The word $\mathtt{baaababb}$ from Examples~\ref{prelims:exampleuniversality} and~\ref{prelims:exampleuniversalitycondt} is not perfect $3$-universal since it has $\iota(\mathtt{baaababb}) = 3$ but $\r(\mathtt{baaababb}) = \tb \neq \varepsilon$. To give a positive example, consider $\mathtt{abcbbaccbbaacb} \in \PUniv_{\{\ta,\tb,\tc\}^*,4}$ since its arch factorisation is  $\mathtt{(abc) \cdot (bbac) \cdot (cbba) \cdot (acb)}$.
\end{example}


\begin{theorem}[\cite{barker2020scattered}]
	Let $w \in \Sigma^{\geq k}$ with $\al(w) = \Sigma$. Then we have $\iota(w)=k$ iff $w$ has exactly $k$ arches.
\end{theorem}

In the remainder of this section, we introduce the basic definitions we need to define the $k$-universality of regular languages. For basic notions on finite automata, we refer to \cite{HopcroftU79}. 
%
%
%
A {\em non-deterministic finite automaton (NFA)} $\mathcal{A}$ is a tuple $(Q,\Sigma,q_0,\delta,F)$ with the finite set of states $Q$ (of cardinality $n\in\N$), an initial state $q_0\in Q$, the set of final states $F\subseteq Q$, an input alphabet $\Sigma$, and a transition function $\delta:Q\times\Sigma\rightarrow 2^{Q}$, where $2^{Q}$ is the powerset of $Q$. If we have $|\delta(q,\ta)|= 1$ for all $q\in Q, \ta\in \Sigma$, then $\mathcal{A}$ is called {\em deterministic} (DFA). We call a sequence $\pi=(q_0,\ta_1,q_1,\ta_2,\ldots,\ta_{\ell},q_\ell)$ an $\ell$-length path in $\mathcal{A}$ iff $q_i\in Q$ for all $i\in[0,\ell]$ and $q_{i+1}\in \delta(q_i,\ta_{i+1})$, for all $i\in[0,\ell-1]$.  
The word $w_{\pi}=\ta_1\cdots\ta_\ell$ is the {\em word (label) associated to $\pi$}.
A path is simple if it does not contain the same state twice. Moreover, a state $q\in Q$ is called accessible (respectively, co-accessible) in $\mathcal{A}$ if there exists a path connecting $q_0$ to $q$ (respectively, $q$ to a final state). A path is called {\em accepting} if $q_{\ell}\in F$ holds. Define the language of $\mathcal{A}$, i.e., the set of words {\em accepted} by $\mathcal{A}$, by $L(\mathcal{A})=\{w\in\Sigma^{\ast}\mid \exists \mbox{ accepting path }\pi\mbox{ in }\mathcal{A}:\,w=w_{\pi}\}$. For abbreviation, we set $\mathcal{A}_n=L(\mathcal{A})\cap\Sigma^n$ and $\mathcal{A}_{\leq n}=L(\mathcal{A})\cap\Sigma^{\leq n}$. Note that the class of languages accepted by NFAs is equal to the class of languages accepted by DFAs and it is equal to the class of regular languages. Moreover, for every word $w\in L(\mathcal{A})$ there exists exactly one (in the deterministic case) or a set of (in the non-deterministic case) associated path(s). For a word $w$ accepted by a finite automaton, let $\pi_w$ denote one accepting path labelled with $w$; in the case when there are multiple such paths, we simply choose one of them. Since we usually are interested in only one path for a word $w\in L(\mathcal{A})$, we refer to it as $\pi_w$.

\begin{definition}\label{def:reverse_transition_function}
	For a given NFA $\mathcal{A}$, the \emph{reverse transition function} $\Delta(q, \tx)$ returns the set of states in $Q$ with a transition labelled by $\tx$ to $q$, i.e, $\Delta(q, \tx) = \{q' \in Q \mid \delta(q', \tx) = q\}$.
\end{definition}

Now, we can define the $k$-subsequence universality of a regular language $L$. Here, we distinguish whether at least one or all words of $L$ are $k$-universal w.r.t.~Definition~\ref{def:subsequence}. Note that we always take the minimal alphabet $\Sigma$ such that $L\subseteq \Sigma^*$ as a reference when considering the $k$-universality of words from $L$.

\begin{definition}
Let $L$ be a regular language. $L$ is called {\em existence $k$-subsequence universal} ($k$-$\exists$-universal) for some $k\in\N$, if there
exists a $k$-universal word $w\in L$. $L$ is called {\em universal $k$-subsequence universal} ($k$-$\forall$-universal) for some $k\in\N$, if all $w\in L$ are $k$-universal. \looseness=-1
\end{definition}

If a regular language $L$, accepted by some finite automaton $\mathcal{A}$ (NFA or DFA), is $k$-$\exists$-universal (respectively, $k$-$\forall$-universal) then we also say that the automaton $\mathcal{A}$ is $k$-$\exists$-universal (respectively, $k$-$\forall$-universal). Further, we also say that a {\em path $\pi$ in $\mathcal{A}$ is $k$-universal} if the label of $\pi$, namely $w_{\pi}$, is $k$-universal. As an example, a $2$-$\exists$-universal NFA $\mathcal{A}$, which is not $3$-$\exists$-universal, and a $1$-$\forall$-universal NFA $\mathcal{B}$ which is not $2$-$\forall$-universal are shown in Figure~\ref{img:universal_automaton} (note that $\mathcal B$ recognises exactly every permutation of $\ta\tb\tc$).
Based on this definition we define two associated decision problems.

\begin{figure}
	\centering
	\begin{minipage}[c]{0.3\textwidth}
		\scalebox{0.7}
		{\begin{tikzpicture}
			[->,>=stealth',shorten >=1pt,auto,node distance=1.6cm,semithick]
			\tikzstyle{every state}=[initial text=$ $]
			\node[initial,state] (A)                    {$q_0$};
			\node[state]         (B) [right of=A]       {$q_1$};
			\node[state,accepting] (D) [right of=B]     {$q_2$};
			
			\path (A) edge [loop above] node {$\mathtt{a}, \mathtt{b}$} (A)
			edge              node {$\mathtt{c}$} (B)
			(B) edge              node {$\mathtt{a}$} (D)
			(D) edge [loop above] node {$\mathtt{b,c}$} (D);
			\node at (-1,1) {$\mathcal{A}$};
		\end{tikzpicture}}
	\end{minipage}	
	\begin{minipage}[c]{0.4\textwidth}
		\scalebox{0.7}
		{\begin{tikzpicture}[->,>=stealth',shorten >=1pt,auto,node distance=1.6cm,
			semithick]
			\tikzstyle{every state}=[initial text=$ $]
			\node[initial,state] (A)                    {$q_0$};
			\node[state]         (B) [below left of=A]       {$q_1$};
			\node[state]         (C) [below left of=B]       {$q_2$};
			\node[state,accepting] (D) [below of=C]       {$q_3$};
			\node[state]         (E) [below of=B]       {$q_4$};
			\node[state,accepting] (F) [below of=E]       {$q_5$};
			\node[state]         (q6) [below of=A]        {$q_6$};
			\node[state]         (q7) [below of=q6]       {$q_7$};
			\node[state,accepting] (q8) [below of=q7]       {$q_8$};
			\node[state]         (q9) [below right of=q6]       {$q_9$};
			\node[state,accepting] (q10) [below of=q9]       {$q_{10}$};
			\node[state]         (q11) [above right of=q9]       {$q_{11}$};
			\node[state]         (q12) [below of=q11]       {$q_{12}$};
			\node[state,accepting] (q13) [below of=q12]       {$q_{13}$};
			\node[state]         (q14) [below right of=q11]       {$q_{14}$};
			\node[state,accepting] (q15) [below of=q14]       {$q_{15}$};
			
			\path (A) edge node {$\mathtt{a}$} (B)
			          edge node {$\mathtt{b}$} (q6)
			          edge node {$\mathtt{c}$} (q11)
 			(B)	edge node {$\mathtt{b}$} (C)
			    edge node {$\mathtt{c}$} (E)
			(C) edge node {$\mathtt{c}$} (D)
			(E) edge node {$\mathtt{b}$} (F)
			(q6) edge node {$\mathtt{a}$} (q7)
			     edge node {$\mathtt{c}$} (q9) 
			(q7) edge node {$\mathtt{c}$} (q8)
			(q9) edge node {$\mathtt{a}$} (q10)
			(q11) edge node {$\mathtt{a}$} (q12)
			      edge node {$\mathtt{b}$} (q14)
			(q12) edge node {$\mathtt{b}$} (q13)
			(q14) edge node {$\mathtt{a}$} (q15);
			\node at (-3,0) {$\mathcal{B}$};
		\end{tikzpicture}}
	\end{minipage}
	\caption{A $2$-$\exists$-universal NFA $\mathcal{A}$ and a $1$-$\forall$-universal  NFA $\mathcal{B}$.}
	\label{img:universal_automaton}
\end{figure}
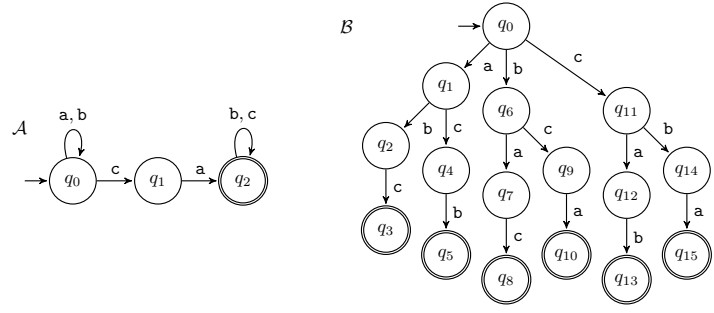

\begin{problem}\label{prob:esubfa}
The {\em existence subsequence universality problem for regular languages ($k$-ESU)} is to decide for a regular language, given by a finite automaton $\mathcal{A}$ recognising it, and $k\in\N$, given in binary representation, whether $L$ is $k$-$\exists$-universal. 
\end{problem}

\begin{problem}\label{prob:asubfa}
The {\em universal subsequence universality problem for finite automata ($k$-ASU)} is to decide for a regular language, given by a finite automaton $\mathcal{A}$ recognising it, and $k\in\N$, given in binary representation, whether $L$ is $k$-$\forall$-universal.
\end{problem}

We finish this section by introducing the {\em rank} of a word in order to enumerate the $k$-universal words accepted by a given finite automaton.

\begin{definition}
	\label{def:rank}
	The \emph{rank} of a word $w$ within the set $\Univ_{\mathcal{M}, k}$ is the number of words in $\Univ_{\mathcal{M}, k}$ lexicographically smaller than $w$, i.e. $rank(w) = \lvert \{ v \in \Univ_{\mathcal{M}, k} \mid v < w\}\rvert$.
\end{definition}

\begin{remark}
	By Definition~\ref{def:subsequence}, 
	the set of all $k$-universal words accepted by $\mathcal{A}$ is given by $\Univ_{L(\mathcal{A}), k}$, the set of all $n$-length  $k$-universal words accepted by $\mathcal{A}$ is given by  $\Univ_{\mathcal{A}_n, k}$, and the set of all $k$-universal words of length at most $n$ accepted by $\mathcal{A}$ is given by $\Univ_{\mathcal{A}_{\leq n}, k}$.
\end{remark}

The computational model we use is the standard unit-cost RAM with logarithmic word size: for an input of size $n$, each memory word can hold $\log (n)$ bits. Arithmetic and bitwise operations with numbers in $[1,n]$ are, thus, assumed to take $O(1)$ time. Numbers larger than $n$, with $\ell$ bits, are represented in $O(\ell/\log n)$ memory words, and working with them takes time proportional to the number of memory words on which they are represented. In all the problems, we assume that we are given a number $k$, binary encoded, and one finite automaton $\mathcal A$, specified as the set of states, set of input letters, set of transitions, and initial and final states. The size of the input is, as such, the size of the binary encoding of $k$, which is $\lceil\log_2 k\rceil$, plus the size $S$ of the encoding of $\mathcal A$ (which is lower bounded by the number of edges in the graph associated to the automaton). So, one memory word can hold $\log (S+\log_2 (k))$ bits.

For a more detailed general discussion on this model see, e.\,g.,~\cite{crochemore}. 

Some of our algorithms run in exponential time and use exponential space w.r.t. the size $n$ of the input. Following the literature dealing with such exponential algorithms (see, e.g., \cite{FominKTV08}), we will use the $O^*$-notation. By definition, for functions $f$ and $g$ we write $f(n) = O^*(g(n))$ if $ f(n)=O(g(n)n^{O(1)})$. In other words, the $O^*$-notation hides polynomial factors, just as the $O$-notation hides constants. Using this notation, we can assume that our single-exponential time and single-exponential space algorithms run on a RAM model where arithmetic and bitwise operations with single exponential numbers (w.r.t. the size $n$ of the input) as well as accessing the memory-words (given their address, like in an usual RAM) are assumed to take $O^*(1)$ time. Working with such a computational model allows us to analyse the actual algorithms rather than the various intricacies of the computational model.

In particular, when expressing the complexity of our algorithms, we get functions which are exponential in $\sigma$ (the size of the alphabet) but polynomial in the number $n$ of states of the input NFA or, for the enumeration algorithms, in the length $m$ of the enumerated strings or in the value of the input number $k$ (the parameter of the considered problems). For clarity of the exposure, although we use the $O^*$-notation, we will explicitly write the dependency on $n, m,$ and $k$, and only hide the polynomial dependency on $\sigma$.


\section{Decision Problems}\label{npc}

In this section, we consider the decision problems $k$-ESU and $k$-ASU. We begin with a series of combinatorial observations.

\begin{remark}
    \label{lem:direct_paths_sufficent_every_word}
     Every path acccepted by  an NFA $\mathcal{A}$ is $k$-universal iff every simple accepting path in it is $k$-universal.
     	In one direction, if $\mathcal{A}$ accepts any path that is not $k$-universal,  then not every path accepted by $\mathcal{A}$ is $k$-universal.
	Otherwise,  the set of paths induced by every word accepted by $\mathcal{A}$ must contain, as a subsequence, some non-cyclic path.
	Therefore, if each non-cyclic path is $k$-universal, then every path is.
\end{remark}

\begin{lemma}
    \label{lem:reaching_every_state_further}
    Let $q, q' \in Q$ be a pair of states in the NFA $\mathcal{A}$ such that there exists a path $\pi$ from $q$ to $q'$ where the final transition in $\pi$ is labelled $\tx \in \Sigma$.
    Then, there exists some path $\pi'$ of length at most $n=|Q|$ from $q$ to $q'$ where the final transition in the path is labelled $\tx$.
\end{lemma}
\begin{proof}
    Let $\hat{q} \in Q$ be the state in $\pi$ before $q'$, i.e., the state such that $q'\in \delta(\hat{q}, \tx) $.
    As there are at most $n$ states in $\mathcal{A}$, given any pair of states $q_1, q_2 \in Q$, if $q_2$ can be reached from $q_1$, then there must exist a path from $q_1$ to $q_2$ of length at most $n - 1$.
    In particular, there exists a path of length $n-1$ from $q$ to $\hat{q}$. Extending this path by the transition 
    from $\hat{q}$ reading $\tx$ to $q'$, we obtain a path $\pi'$ from $q$ to $q'$ of length at most $n$.
\end{proof}

\begin{lemma}
	\label{lem:ksu_min_length}
	For an NFA $\mathcal{A}$ with $n$ states, if there is a $k$-universal word accepted by $\mathcal A$, then there is a $k$-universal word accepted by $\mathcal{A}$ of length at most $ kn\sigma - (n-1)(k-1)$. 
\end{lemma}
\begin{proof}
Let $w$ be a $k$-universal word accepted by $\mathcal{A}$ and for all $i\in [k]$ let $a_{i,1},\ldots a_{i,\sigma}\in\Sigma$ be the letters of $\Sigma$ in the order they occur in $\ar_i(w)$, that is $a_{i,1}=\ar_i(w)[1]$ and $a_{i,j} =\ar_i(w)[\ell]$ such that $j-1 = |\al(\ar_i(w)[1,\ell-1])|<|\al(\ar_i(w)[1,\ell])|=j$ for $1<j\leq\sigma$.
Then $\ar_i(w) = a_{i,1}u_{i,1}a_{i,2}\cdots u_{i,\sigma-1}a_{i,\sigma}$ where $u_{i,j}\in\Sigma^\ast$ is a word such that there is a path $\pi_{u_{i,j}}$ in $\mathcal A$ labeled with $u_{i,j}$. For the sake of readability we denote the suffix of $w$ starting after $\ar_k(w)$ by $u_{k+1,1}$ in this proof. By \Cref{lem:reaching_every_state_further}, $u_{i,j}$ is accepted by an automaton with at most $n$ states obtainable from $\mathcal A$ by changing the initial (respectively, final) state of $\mathcal{A}$ to the starting (respectively, end) state of $\pi_{u_{i,j}}$. 
Hence we can find a word $v_{i,j}$ of length at most $n-1$ which is the label of a path starting and ending in the same state as $\pi_{u_{i,j}}$. 
Let $w'$ be the word obtained from $w$ by replacing every $u_{i,j}$ by $v_{i,j}$. Then $w'$ is a $k$-universal word accepted by $\mathcal A$ of length $\lvert w'\rvert = \sum_{i=1}^{k}\left(\sum_{j=1}^\sigma|a_{i,j}| + \sum_{j=1}^{\sigma-1}|v_{i,j}|\right) + |v_{k+1,1}| \leq k\left(\sigma + (\sigma-1)(n-1)\right) + n-1 = kn\sigma - (k-1)(n-1)$ (cf. Figure~\ref{img:worstcaseexample}).
\end{proof} 

    \begin{figure}
	\scalebox{0.8}
	{\begin{tikzpicture}[->,>=stealth',shorten >=1pt,auto,node distance=2.5cm,
		semithick]
		\tikzstyle{every state}=[initial text=$ $]
		
		\node[state,initial] (q1) {$q_{1}$};
		\node[state] (q2) [right of=q1] {$q_{2}$};
		\node[state] (q3) [right of=q2] {$q_{3}$};
		\node[state,accepting] (q4) [right of=q3] {$q_{4}$};
		
		\path 
		(q1) edge node {$\mathtt{a}$} (q2)
		(q2) edge node {$\mathtt{a}$} (q3)
		(q3) edge node {$\mathtt{a}$} (q4)
		(q4) edge [bend right,above] node {$\Sigma \setminus \{\mathtt{a}\}$} (q1);
		\node at (-1.5,1) {$\mathcal{A}$};
	\end{tikzpicture}}
	\caption{Note that the shortest $k$-universal word accepted by $\mathcal{A}$ has length $(\sigma - 1) k n$.}
	\label{img:worstcaseexample}
\end{figure}
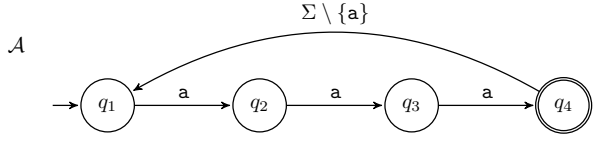

We now move to the main results of this section and provide algorithms for solving the $k$-ESU and $k$-ASU problems. First,  we show that $k$-ASU can be solved in $O(n^3 \sigma)$ time. Secondly, we provide a fixed parameter tractable (FPT) algorithm w.r.t.~$\sigma$ for solving $k$-ESU. The time complexity of this algorithms is polynomial for $\sigma\in O(\log n)$ and does not depend at all on $k$. We complete this section by showing that $k$-ESU is NP-complete in the general case, justifying the need for an FPT algorithm.




\begin{theorem}
	For a given NFA $\mathcal{A}$ with $n$ states and $|\Sigma|=\sigma$ and a natural number $k\in\N$, we can decide $k$-ASU in $O(n^3 \sigma)$ time.
\end{theorem}

\begin{proof}
	Let $\mathcal A=(Q,\Sigma,q_0,F,\delta)$ be an NFA with $|Q|=n$. 
	If $k>n$ then there exists at least one word $w\in L(\mathcal{A})$ with $|w|<k$ and $L(\mathcal{A})$ is not $k$-universal. Thus, assume $k\leq n$. The main idea is to check, whether there exists $w\in L(\mathcal{A})$ with $\iota(w)=\ell<k$ since, if so, we know that not all words accepted by $\mathcal{A}$ are $k$-universal.
	Let $\mathcal{P}(q,q')$ be the set of all paths from $q$ to $q'$ and let $\letters(\pi)$ be the set of all (letter-)labels
	of the edges appearing on a path $\pi$ in $\mathcal{A}$.
	
	Define for all $\ta\in\Sigma$ the relation
	\[
	R_{\ta}=\{(q,q')\in Q\times Q|\, \exists q''\in Q, \exists \pi\in\mathcal{P}(q,q''):\letters(\pi)\subseteq \Sigma\backslash\{\ta\}\land q'\in\delta(q'',\ta)\},
	\]
	i.e., $(q,q')\in R_{\ta}$ iff there exists a path from $q$ to $q'$ where exactly only the last transition is labelled with $\ta$ (while all the others are labelled with other letters).
	
	Note that for a fixed $\ta\in\Sigma$, $R_\ta$ can be constructed in polynomial time: we first take $\mathcal A$ and eliminate all transitions labelled with $\ta$ in $O(n^2)$ time, in the worst case (as ${\mathcal A}$ is an NFA). Then, in $O(n^3)$ time we can determine all pairs of states $(q_1,q_2)$ such that there is a path in this modified automaton from $q_1$ to $q_2$ containing only letters of $\Sigma\backslash\{\ta\}$ as labels. Then, for each pair $(q_1,q_2)$ connected by a path in this modified automaton, we follow one edge labelled with $\ta$ from $q_2$, and identify, in this way, the pair $(q_1,q_3)$ such that $(q_1,q_3)\in R_\ta$. Set $R=\bigcup_{\ta\in\Sigma}R_{\ta}$. Overall, this entire process can be implemented in time $O(n^3 \sigma)$.
	
	Define now the set 
	\[
	Q'=\{q\in Q | \exists q'\in F, \exists \pi\in\mathcal{P}(q,q'):\,\letters(\pi)\subsetneq \Sigma\}. 
	\]
	This can be computed, trivially, in $O(n^2 \sigma)$ overall, using the relations $R_{\ta}$, for $\ta\in\Sigma$, defined above. 
	
	Further, define a directed graph $G$, which has the set of vertices equal to $Q$, and there is an edge from $q$ to $q'$ if and only if $(q,q')\in R$.
	
	Now, we show that there exists a word in $L(\mathcal{A})$ with $\iota(w)=\ell$ if and only if the shortest path from $q_0$ to a state from $Q'$, in the graph $G$, is of length $\ell$.
	
	If there exists a word $w\in L(\mathcal{A})$ with $\iota(w)=\ell$, then, using the definition of the arch factorisation, $w=w_1\ta_1w_2\ta_2\cdots w_\ell\ta_\ell w_{\ell+1}$, where, for all $i\leq \ell+1$, we have $\iota(w_i)=0$, and, for all $i\leq \ell$,  we have $\iota(w_i\ta_i)=1$. Now, this means, clearly, that there are states $q_1,\ldots, q_\ell$ such that $(q_{i-1},q_i)\in R_{\ta_i}$, for all $i\leq \ell$, and $q_\ell$ is in $Q'$ (being connected to a final state by a path labelled by the word $w_{\ell+1}$). So, there exists a path of length $\ell$ from $q_0$ to $q_\ell\in Q'$ in the graph $G$. 
	
	For the converse, if the shortest path from $q_0$ to a state from $Q'$, in the graph $G$, has length $\ell$, then there exists a sequence of states $q_1,\ldots, q_\ell$ such that $(q_{i-1},q_i)\in R_{\ta_i}$ for some $\ta_i\in \Sigma$, for all $i\leq \ell$, and $q_\ell$ is in $Q'$. Thus, for all $1\leq i\leq \ell$, there exists a word $w_i$ which labels a path from $q_{i-1}$ to $q_i$, such that $w_i$ does not contain the letter $\ta_i$, and a word $w_\ell+1$ which connects $q_\ell$ to a final state of ${\mathcal A}$ and $\iota(w_{\ell+1})=0$. Therefore, we immediately obtain that the word $w=w_1\ta_1\ldots w_\ell\ta_\ell w_{\ell+1}$ is in the language accepted by ${\mathcal A}$. Moreover, as $\iota(w_i\ta_i)\leq 1$, for all $1\leq i\leq \ell$ and $\iota(w_{ell+1})=0$, we get that $\iota(w)\leq \ell$. However, if $\iota(w)<\ell$, by the implication above we would obtain that the shortest path in ${\mathcal A}$ has length strictly smaller than $\ell$. Thus, $\iota(w)=\ell$. 
	
	So, we can conclude that the minimum $\ell$ for which there exists a word $w$ accepted by ${\mathcal A}$ with $\iota(w)=\ell$ can be determined by finding the shortest path between $q_0$ and a state of $Q'$, in the graph $G$. This step can clearly be implemented in $O(n^3 \sigma)$ time. To answer the instance of $k$-ASU defined by the NFA ${\mathcal A}$ and the integer $k$, it is enough to check whether $k$ is greater or equal to the length of this minimum $\ell$ for which there exists a word $w$ accepted by ${\mathcal A}$ with $\iota(w)=\ell$. 
	
	This concludes our proof.
\end{proof}

Worth noting, the proof of the previous theorem states that, for a given NFA ${\mathcal A}$, we can compute in polynomial time the smallest natural number $\ell$ for which there exists a word $w$ accepted by ${\mathcal A}$ with $\iota(w)=\ell$. This seems an interesting result as such, not only in the context of the previous theorem.

We now consider the $k$-ESU problem.

\begin{lemma}\label{lem:ESUPFA}
For a given NFA $\mathcal{A}$ with $n$ states and $|\Sigma|=\sigma$, we can decide in $O^*(n^3 2^\sigma)$ time whether $\mathcal{A}$ accepts words whose universality index is arbitrarily large. If the answer is negative, then we can compute the largest universality index of a word accepted by $\mathcal{A}$. 
\end{lemma}

\begin{proof}
Assume that $\mathcal{A}=(Q,\Sigma,q_0,\delta,F)$, as defined in Section \ref{prelims}. Let us assume w.l.o.g.~that $\mathcal{A}$ contains only accessible and co-accessible states.

Our approach is based on a series of observations. We first make these observations, and then provide an algorithm proving the statement of the lemma.

Note first that if there exists a state $q\in Q$ (which is both accessible and co-accessible, by assumption) such that there is a path in $\mathcal{A}$ from $q$ to $q$ labelled with a word which contains all letters of $\Sigma$, then we can immediately decide that  $\mathcal{A}$ accepts words whose universality index is arbitrarily large. Indeed, to obtain an accepted word (accepting path), whose universality index is at least $\ell$, we follow the path from $q_0$ to $q$, then follow $\ell$ times the path which contains all letters of $\Sigma$ going from $q$ to $q$, and finally follow a path connecting $q$ to a final state.

Secondly, assume that there exists no state $q\in Q$ such that there is a path in $\mathcal{A}$ from $q$ to $q$ labelled with a word which contains all letters of $\Sigma$. In this setting we note that, for each state $q\in Q$, there exists a unique maximal (w.r.t. inclusion) subset $V_q\subsetneq \Sigma$ such that there exists a path from $q$ to $q$ labelled with a word $\beta_q$ with $\letters(\beta_q)=V_q$. Indeed, if two such different maximal sets $V'_q$ and $V''_q$ would exist, witnessed by the words $w'$ and $w''$, then we can follow the path labelled with $w'w''$, which goes from $q$ to $q$, and $\letters(w'w'')=V'_q\cup V''_q$ (and this includes both $V'_q$ and $V''_q$). 

Finally, consider an accepting path $\pi=(q_0,\ta_1,q_1,\ta_2,\ldots,\ta_{\ell},q_\ell)$ of $\mathcal{A}$. We can rewrite the path $\pi$ as follows:
\begin{itemize}
\item Find the rightmost occurrence of $q_0$ in $\pi$; let us assume that this is the \nth{$h$} state on this path, namely $q_h$. Replace the subpath connecting the initial occurrence of $q_0$ to $q_h$ with the loop $(q_0, \ta_1\cdots\ta_h)^{\circ}$ (where we use $\circ$ to emphasise this is a loop). Now the path is rewritten as $((q_0, \ta_1\cdots\ta_h)^\circ, \ta_{h+1},q_{h+1},\ldots,\ta_{\ell},q_\ell)$. If $q_0$ appears only once in $\pi$, then the path is rewritten as $((q_0, \varepsilon)^\circ, \ta_{1},q_{1},\ldots,\ta_{\ell},q_\ell)$. 
\item Now, we repeat the procedure for the path $(q_{h+1},\ta_{h+2},\ldots,\ta_{\ell},q_\ell)$ (respectively, if $q_0$ appeared only once on $\pi$, for the path $(q_{1},\ldots,\ta_{\ell},q_\ell)$). 
\end{itemize}
After completing this process, one obtains a {\em normal-form representation} of the path $\pi$ as $((q'_0, \alpha_1)^\circ, \ta'_{1}, (q'_{1},\alpha_2)^\circ, \ta'_2, \ldots,\ta'_{r},(q'_r,\alpha_{r+1})^\circ )$, with $r\leq n$, as the states $q'_0=q_0, q'_1, \ldots, q'_r$ are pairwise distinct; moreover, $\alpha_i\in \Sigma^*$ for all $i\in [r]$. 

Let us come back now to the case when there exists no state $q\in Q$ such that there is a path in $\mathcal{A}$ from $q$ to $q$ labelled with a word which contains all letters of $\Sigma$. 
Consider now all accepting paths $((q'_0, \alpha_1)^\circ, \ta'_{1}, (q'_{1},\alpha_2)^\circ, \ta'_2, \ldots,\ta'_{r},(q'_r,\alpha_{r+1})^\circ )$ (given in normal-form representation) where the states $q'_0,\ldots, q'_r$ and the letters $\ta'_1,\ldots,\ta'_r$ are fixed, and the loops $\alpha_1,\ldots,\alpha_r$ are variable. We are interested in how we can, for $i\in [r+1]$, choose $\alpha_i\in V_{q_{i-1}}^\ast$ in order to maximise the universality index of the resulting word. 
We claim that it is enough to take $\alpha_i=\beta^2_{q_{i-1}}$ (where the words $\beta_q$ were defined above). Indeed, this holds because the arch decomposition of the word labelling the path $((q'_0, \alpha_1)^\circ, \ta'_{1}, (q'_{1},\alpha_2)^\circ, \ta'_2, \ldots,\ta'_{r},(q'_r,\alpha_{r+1})^\circ )$ is done greedily from left to right, and we are thus interested in packing as many arches as possible in each of the prefixes $((q'_0, \alpha_1)^\circ, \ta'_{1}, (q'_{1},\alpha_2)^\circ, \ta'_2, \ldots,\ta'_{i},(q'_i,\alpha_{i+1})^\circ )$, for all $i$. 
So, we begin by noting that the alphabet of the word labelling $((q'_0, \alpha_1)^\circ)$ is always included in $V_{q'_0}$, and is indeed equal to $V_{q'_0}$ for, e.g., $\alpha_1=\beta^2_{q'_0}$. 
Now, let us consider $i\geq 1$ and the path $((q'_0, \alpha_1)^\circ, \ta'_{1}, (q'_{1},\alpha_2)^\circ, \ta'_2, \ldots,\ta'_{i},(q'_i,\alpha_{i+1})^\circ )$. 
We note that, as we do not have $1$-universal loops, the loop $(q'_i,\alpha_{i+1})^\circ$ can at most complete the final (and previously incomplete) arch of $((q'_0, \alpha_1)^\circ, \ta'_{1}, (q'_{1},\alpha_2)^\circ, \ta'_2, \ldots,\ta'_{i-1},(q'_{i-1},\alpha_{i})^\circ )$; after this potential completion, we can only try to add as many new letters as we can in the last arch. So, we will define $\alpha_{i+1}$ as a power of the word $\beta_{q'_{i}}$, that labels the loop containing $q'_{i}$ which adds the most new letters to the constructed word, and moreover it is enough to only use $\beta_{q'_i}$ twice in this power (i.e., $\alpha_{i+1}=\beta^2_{q_{i}}$). 
To see that this holds, note that if the rest of the word labelling the path $((q'_0, \alpha_1)^\circ, \ta'_{1}, (q'_{1},\alpha_2)^\circ, \ta'_2, \ldots,\ta'_{i-1},(q'_{i-1},\alpha_{i})^\circ )$ is completed to a new arch by a repetition $\beta^f_{q'_i}$ for some $f>2$, then this is done by the first $\beta_{q'_i}$ from this repetition. Then, the suffix $\beta^{f-1}_{q'_i}$ adds as many new letters to the alphabet of the rest of the resulting word as a single occurrence of  $\beta_{q'_i}$. This completes the proof of our claim. 

In conclusion: among all accepting paths $((q'_0, \alpha_1)^\circ, \ta'_{1}, (q'_{1},\alpha_2)^\circ, \ta'_2, \ldots,\ta'_{r},(q'_r,\alpha_{r+1})^\circ )$ (given in normal-form representation) the path $((q'_0, \beta^2_{q'_0})^\circ, \ta'_{1}, (q'_{1},\beta^2_{q'_1})^\circ, \ta'_2, \ldots,\ta'_{r},(q'_r,\beta^2_{q'_r})^\circ )$ has the highest universality index.

Based on these observations, we can now state the main idea of our algorithm. We first preprocess the input automaton to get rid of the states which are not accessible and of the states which are not co-accessible. Then, we check whether there exists a state $q\in Q$ such that there is a path in $\mathcal{A}$ from $q$ to $q$ labelled with a word which contains all letters of $\Sigma$; if yes, we simply decide that there exists $k$-universal accepted words, for any $k$. If no such state exists, we compute, by dynamic programming, the path $((q'_0, \beta^2_{q'_0})^\circ, \ta'_{1}, (q'_{1},\beta^2_{q'_1})^\circ, \ta'_2, \ldots,\ta'_{r},(q'_r,\beta^2_{q'_r})^\circ )$ with the highest universality index, for all states $q_r\in Q$. 

In the following we describe this algorithm in more details, by highlighting its main phases.

{\em Preprocessing:} We determine the accessible states of $\mathcal{A}$ by running a graph traversal algorithm on $\mathcal{A}$ starting from the state $q_0$. The not-accessible states are removed. We determine the co-accessible states of $\mathcal{A}$ by running a graph traversal algorithm on the graph of $\mathcal{A}$, with the direction of all edges inverted, starting from the final states. The not-co-accessible states of $\mathcal{A}$ are removed. The complexity of this step is $O(n^3)$ in the worst case. Note that for this step it is enough to consider a simplified form of the graph associated to $\mathcal{A}$, where we only see, for each two states, if there is at least one edge between them or not; as such, the total size of the graph $\mathcal{A}$ is $O(n^2)$). From now on, we assume that all states of $\mathcal{A}$ are both accessible and co-accessible.

{\em Universal Loops:} In this step, we determine whether there exists a state $q$ of $\mathcal{A}$ such that there is a path in $\mathcal{A}$ from $q$ to $q$ labelled with a word which contains all letters of $\Sigma$. This is done as follows. For each state $q\in Q$, we run the following process. We define $L$ to be a queue, initially containing the element $(q,\emptyset)$, and $S$ be a set which is initially empty. As long as $L$ is not empty, we pop the first element from $L$, let it be $(q',R')$, and update $S$ to $S\cup\{(q',R')\}$. Now, for all transitions leaving $q$ of the form $q''\in \delta(q',\ta)$, if $(q'',R'\cup\{\ta\})\not\in S$, we insert $(q'',R'\cup\{\ta\})$ into $L$. When $L$ is empty, we are done, and we simply check if $(q,\Sigma)$ is contained in $S$. If $(q,\Sigma)$ is contained in $S$ then there exists a path from $q$ to $q$ which is $1$-universal. Otherwise, there is no such path. Moreover, we also compute and store the set $V_q$ of maximal cardinality for which $(q,V_q)$ belongs to $S$.

After running the above process for all states in $Q$, if we have identified a state $q$ for which there exists a path from $q$ to $q$ which is $1$-universal, then we simply conclude that $\mathcal{A}$ accepts words whose universality index is arbitrarily large. Otherwise, we continue with the next phase.

The complexity of this phase is $O^*(n^3 2^{\sigma})$, as the numbers of transitions leaving a state $q$ is upper bounded by $n\sigma$. 

{\em No Universal Loop:} In this case, we will use a $n\times 2^\sigma$ matrix $M[\cdot][\cdot]$ (and an auxiliary matrix $M'$ of the same size). Our algorithm has at most $n$ iterations, and we will maintain the property that, before iteration $\ell+1$ (for $\ell\geq 0$),  $M[q'_r][V]$ is the maximum among the number of arches of some path with the normal-form $((q'_0, \beta^2_{q'_0})^\circ, \ta'_{1}, (q'_{1},\beta^2_{q'_1})^\circ, \ta'_2, \ldots,\ta'_{r},(q'_r,\beta^2_{q'_r})^\circ )$, with $q'_0=q_0$ and $r\leq \ell$, such that the rest $\r(u)$ of the word $u$ labelling this path has the alphabet $V$, or  $M[q'_r][V]=-1$ if no such path exists.

To begin with, before the first iteration of the loop, we initialize $M[q][V]=-1$ for all $q$ and $V$. Then we set $M[q_0][V_{q_0}]=0$. 

In the \nth{$\ell$} iteration, we copy $M$ into the matrix $M'$. We then go through all $q$ and $V$ and, if $M[q][V]\neq -1$, we do the following: 
for all letters $a\in \Sigma$, for all $q'\in \delta(q,a)$, set $V'=V\cup \{a\}$ and $\delta=0$. If $V'=\Sigma$, we reset $V'=\emptyset $ and set $\delta=1$. 
Then, we compute $V''=V'\cup V_{q'}$. If $V''=\Sigma$ (note that this can only happen when $V'\neq \emptyset$) we reset $V''=\emptyset $ and set $\delta=1$. Finally, we compute $V'''=V''\cup V_{q'}$ and set $M'[q][V''']=\max\{M'[q'][V''],M[q][V]+\delta\}$.

We stop this process after $n$ iterations.

Clearly, before the first iteration of our algorithm, the property we intend to maintain holds. Then, in the \nth{$i$} iteration, we try to extend the already constructed path ending in state $q$, with rest $V$, by reading first one letter and reaching state $q'$, and then by reading $\beta^2_{q'}$, and compute the universality index of this resulting path. So, the property is maintained in the \nth{$\ell$} iteration. 

The complexity of this phase is $O^*(n^3  2^\sigma)$, if all the above steps are implemented na\"ively. 

Now, the highest universality index of a word accepted by $\mathcal{A}$ is the largest value $M[q][V]$ for $q\in F$, $V\subseteq \Sigma$.

{\em Conclusion:} The algorithm consisting of the three phases above decides whether $\mathcal{A}$ accepts words whose universality index is arbitrarily large, or, if this is not the case, computes the largest universality index of a word accepted by $\mathcal{A}$. Its time complexity is $O^*(n^3 2^\sigma)$.
\end{proof}

As a corollary of Lemma \ref{lem:ESUPFA}, we get the following theorem.

\begin{theorem}\label{thm:ESUPFA}
For a given NFA $\mathcal{A}$ with $n$ states and $|\Sigma|=\sigma$ and a natural number $k\in\N$, we can decide $k$-ESU in $O^*(n^3 2^\sigma)$ time.
\end{theorem}

\begin{proof}
We first use Lemma \ref{lem:ESUPFA} to test whether $\mathcal{A}$ accepts words whose universality index is arbitrarily large. If yes, then the answer for the given instance of $k$-ESU is positive. If $\mathcal{A}$ does not accept words whose universality index is arbitrarily large, we compute the largest universality index $\ell$ of a word accepted by $\mathcal{A}$. If $\ell\geq k$, then the answer for the given instance of $k$-ESU is positive. Otherwise, the answer for the given instance of $k$-ESU is negative.
\end{proof}

%
%
%
%
%

We conclude this section by showing that $k$-ESU is actually NP-complete, and it is NP-hard even for $k=1$. Clearly, in the light of our previous results from Theorem~\ref{thm:ESUPFA}, for this problem to be NP-hard we need to consider the case of an input alphabet $\sigma\in \Omega(\log n)$. 

\begin{remark}
Note that an automaton accepts a word $w$ which is $k$-universal for $k>n = \vert Q \vert$ iff there exists a state $q$, which is both accessible and co-accessible, and a loop from $q$ to $q$ labelled with a $\ell$-universal word, for some $\ell\geq 1$. Indeed, for the 
	left-to-right implication, we look at the states $q_i$ reached by reading the first $i$ arches of the word $w$, for $i\in [k]$. Clearly, there will be some $i<j$ such that $q_i=q_j$, and the subpath between $q_i$ and $q_j$ on the path labelled with $w$ is $\ell$-universal for some $\ell\geq 1$. The other implication is immediate.
\end{remark}

\begin{theorem} \label{thm:NP_hard_1_universal}
	$k$-ESU is NP-complete.
\end{theorem}
\begin{proof}
	We begin by showing that this problem is in NP.

	To solve $k$-ESU in NP-time, we first check if the automaton $\mathcal A$ contains a state $q$, which is both accessible and co-accessible, and a path from $q$ to $q$ labelled with a $x$-universal word. First, we traverse (deterministically) the graph of $\mathcal{A}$ from $q_0$ to determine the accessible states, and also from the final states using the inverted edges to determine the co-accessible states. For each of these states, we non-deterministically guess a $1$-universal word of length at most $n\sigma$ and see if this takes us from $q$ to $q$ (by Lemma~\ref{lem:ksu_min_length} we have that, if there exists a $1$-universal word accepted by an altered version of the automaton $\mathcal{A}$, where $q$ is the unique initial and final state, then there exists one such word of length at most $n\sigma$). If we successfully guessed this word, and $k > n$, then we simply accept the input automaton; otherwise, we reject. If $k \leq n$ and no such altered autonaton exists, then follwoing Lemma \ref{lem:ksu_min_length} any $k$-subsequence universal word must have length at most $k n \sigma$. Therefore, we can check every word in $\Sigma^{k n \sigma}$ to determine if any word is both $k$-subsequence universal and accepted by $\mathcal{A}$. If such a word is found we accept the input automaton, otherwise we reject.
	
	The fact that the problem is NP-hard follows from the proof of Theorem 3 from \cite{KimHKS22}. For completeness (as the respective proof is not given in the accessible version of the paper), we sketch here a reduction from the Hamiltonian Path Problem. The high level idea behind this reduction is to take a graph $G = (V, E)$ containing $n$ vertices $v_1,\ldots, v_n$ and construct an automaton $\mathcal{A}$ containing $n^2 + 2$ states, with a unique starting state $q_0$, a unique failure state $q_f$, and a set of $n^2$ states labeled $q_{i, j}$ for every $i, j \in [1, n]$.
The state  $q_{i, j}$ is used to represent visiting the vertex $v_j$ at the $i^{th}$ step of some path in $G$. 
With this in mind, a transition exists from $q_{i, j}$ to $q_{\ell, k}$ if and only if $\ell = i + 1$, and $(j, k)$ is an edge in $G$.
In order to map the paths accepted by this automaton directly to the paths in $G$, each transition is labelled by the index $k$ corresponding to the end state of the edge, i.e., the transition between $q_{i, j}$ and $q_{i + 1, k}$ is labelled by $k$.
With this construction, the path in $\mathcal{A}$ labelled by the word $w$ corresponds directly to some path of length $\vert w \vert$ in $G$. 

As a Hamiltonian path must have length exactly $n$, for every $j \in [1, n]$ the state $q_{n, j}$ is marked as an accepting state, and every transition from $q_{n, j}$ leads to the failure state $q_{f}$.
From this construction, any $1$-universal word must correspond exactly to some permutation of the alphabet $[n]$, representing a Hamiltonian path in $G$.
In the other direction if no such word exists, then there does not exist any such path.
\end{proof}

In conclusion, we showed that $k$-ASU is solvable in polynomial time whereas $k$-ESU is NP-complete.

\section{Counting and Ranking}
\label{sec:counting}


In this section we discuss the problems of counting and ranking efficiently $k$-universal words from a regular language, given as a DFA, or $k$-universal paths in the case when the respective language is given as an NFA. Note, that there is a one-to-one correspondence between paths and words in the case of DFAs, so, for the sake of simplicity, from now on we will simply talk about counting and ranking paths accepted by the finite automata we are given as input.

Here, we define the problem of \emph{counting} and \emph{ranking} problems, for a given automaton $\mathcal{A}$, and universality index $k$. The counting problem is defined as the problem of determining the number of $k$-universal words accepted by $\mathcal{A}$, equivalent to determining the size of $\Univ_{L(\mathcal{A}), k}$. For a given length $m \in \mathbb{N}$, the problem of counting the number of $k$-universal words of length exactly (respectively, at most) $m$ is the problem of determining the size of $\Univ_{\mathcal{A}_m, k}$ (respectively $\Univ_{\mathcal{A}_{\leq m}, k}$. The \emph{rank} of a word $w \in \mathcal{A}_n$ is the number of $k$-universal words accepted by $\mathcal{A}$ that are smaller than $w$, i.e. the size of  $\{v < w \mid v \in \Univ_{L(\mathcal{A}), k} \} $. For a given length $m$, the \emph{rank} of $w$ within the set of $k$-universal words of length exactly (respectively, at most) $m$ is the problem of determining the size of the set $\vert \{ v < w \mid v \in \Univ_{\mathcal{A}_m, k} \} \vert$ (respectively, $\vert \{ v < w \mid v \in \Univ_{\mathcal{A}_{\leq m}, k} \} \vert$).

Recall that we are operating on a DFA or NFA $\mathcal{A}$ with $n$ states and an alphabet of size $\sigma > 1$ (the case $\sigma=1$ is trivial). We are also given as input the natural numbers $k$ and $m$ in binary representation: we are interested in counting (ranking) the $k$-universal words of length $m$ contained in $L(\mathcal{A})$. It is important to know that these problems are only interesting for $k\leq m/\sigma$; otherwise, there are no $m$-length $k$-universal words. However, an additional difficulty related to this problem, compared to the case of the decision problems discussed in Section \ref{npc}, is that we need to do arithmetics with large numbers. In general, if $M$ is an upper bound on the value of the numbers that we need to process and $\omega$ is the size of the memory word of our model, then each arithmetic operation requires at most $O(\frac{ \log M }{\omega})$ time to complete. In the cases we approach here, $M\leq (n\sigma)^{m+1}$ (a crude upper bound on the total number of paths of length at most $m$ in $\mathcal{A}$), so each arithmetic operation requires at most $O(\frac{(m+1) \log (n\sigma) }{\omega})$, that is $O^*(m)$ time. 

This section is laid out as follows. First, we consider the problems of counting the number of $k$-universal accepting  paths of the finite automaton $\mathcal{A}$. For a given $m \in \N$ this is split into two cases, counting  the number of $k$-universal accepting paths of $\mathcal{A}$ of length exactly $m$, and counting the number of $k$-universal accepting paths of $\mathcal{A}$ of length at most $m$. We show that both can be computed in $O^*(m^2 n^2 k 2^{\sigma})$ time.
Additionally, we show that the number of $k$-universal accepting paths of $\mathcal{A}$ can be computed in 
$O^*(n^4 k^2 2^{\sigma})$ time.

We extend our counting results to the ranking setting, showing that a path $\pi$ (respectively, word $w$) can be ranked within the set of $k$-universal paths (respectively, words) of length exactly $m$ accepted by the NFA (respectively, DFA) ${\mathcal A}$ in 
$O^*(m^2 n^2 k 2^{\sigma})$ 
time, of length at most $m$ in $O^*(m^2 n^2 k 2^{\sigma})$ time, and of any length in $O^*(n^4 k^2 2^{\sigma})$ time.

The main tool used in this section is the $n \times (m + 1) \times k \times 2^{\sigma}$ size table $T$, which we refer to as the \emph{path table of length $m$} for a given $m\in\N_0$.
Each entry in the table $T$ is indexed by a state $q \in Q$, a length $\ell \in [0, m]$, the number of arches $c \in [0, k-1]$, and a subset of symbols $\mathcal{R} \subset \Sigma$. The entry $T[q, \ell, c, \mathcal{R}]$ contains the number of $\ell$-length paths starting at the state $q_0$ and ending in the state $q$ such that the words induced by the paths contain each $c$ arches, and the alphabet of the rest  is $\mathcal{R}$. Thus,
in the context of the arch factorisation, we are interested in all words belonging to paths in $\mathcal{A}$ which are {\em not yet} $k$-universal.
Formally, we have $T[q, \ell, c, \mathcal{R}] = \sum_{w \in \Sigma^\ell, \letters(\r(w)) = \mathcal R, \iota(w)=c} \vert \mathcal{P}(w, q) \vert$, where $\mathcal{P}(w, q)$ is the set of paths from $q_0$ to $q$ in $\mathcal{A}$ labelled by $w$.
Note, that for DFAs, we have $|\mathcal{P}(w,q)|=1$.
The words which have at least $k$ arches in their arch factorisation are captured in the auxiliary $n \times m + 1$ size table $U[q, \ell]$, which we refer to as the \emph{universal words table}. Each entry in $U$ is indexed by a state $q \in Q$ and length $\ell \in [0, m]$ with the entry $U[q, \ell]$ containing the number of $\ell$-length paths ending at $q$ which are $k$-universal, i.e., $\iota(w)\geq k$.


The remainder of this section provides the combinatorial and technical tools needed to construct the tables $T$ and $U$.
Theorems \ref{thm:counting_exactly_ell_k_subsequence_universal_words} and \ref{thm:counting_everything}, and Corollary \ref{col:counting_up_to_ell_k_subsequence_univeral_words} summarise the main complexity results of this section.

\begin{lemma}
    \label{lem:making_a_path_longer}
    Let $\pi$ be an $(\ell - 1)$-length path in the NFA $\mathcal{A}$ ending at $q$ and corresponding to the word $w_{\pi}$, such that $\iota(w_{\pi}) = c$ and $\letters(\r(w_{\pi})) = \mathcal{R}$. 
    Then the word $w_{\pi'}$ corresponding to the path $\pi'$ formed by following a transition labelled $\tx\in\Sigma$ from $q$ either:
    \begin{itemize}
        \item has an empty rest $(\r(w_{\pi'}) = \varepsilon)$ if $\mathcal{R} \cup \{\tx\} = \Sigma$ and hence $\iota(w_{\pi'}) = c + 1$, or
        \item has a rest equal to $\mathcal{R} \cup \{ \tx \}$ $(\letters(\r(w_{\pi'})) =  \mathcal{R} \cup \{\tx\})$ if $\mathcal{R} \cup \{\tx\} \subsetneq \Sigma$ and hence $\iota(w_{\pi'}) = c$. 
    \end{itemize}
\end{lemma}
\begin{proof}
	In the first case, if $\letters(\r(w_{\pi})) = \Sigma \setminus \{ \tx \}$ then $\r(w_{\pi}) \tx$ is an arch, as it contains every symbol from $\Sigma$ at least once, and hence the arch factorisation of $w_{\pi'}$ contains $c + 1$ arches, and an empty rest, i.e. $\iota(w_{\pi'}) = c + 1$ and $\r(w_{\pi'}) = \varepsilon$.
	Otherwise, $\r(w_{\pi}) \tx$ contains the set of letters $\mathcal{R} \cup \{\tx\}$, and does not complete a new arch.
	Hence the arch factorisation of $w_{\pi'}$ contains $c$ arches and the rest contains the letters $\mathcal{R} \cup \{\tx\}$, i.e., $\iota(w_{\pi'}) = c$ and $\al(\r(w_{\pi'})) = \al(w_{\pi}) \cup \{ \tx \}$.
\end{proof}

Lemma \ref{lem:making_a_path_longer} provides the outline of the dynamic programming approach used to compute the table $T$.
Starting with the $0$-length path corresponding to the empty word $\varepsilon$, the value of $T[q, \ell, c, \mathcal{R}]$ is computed from the values of $T[q', \ell - 1, c', \mathcal{R}']$, allowing an efficient computation of the table.
Corollary \ref{lem:path_prefixes} rewrites this in terms of computing the number of paths of length $\ell$ ending at state $q$ corresponding to words with $c$ arches and the set of symbols $\mathcal{R}$ of $w$'s rest. The proof is analogous to the one of Lemma~\ref{lem:making_a_path_longer}.

\begin{corollary}
    \label{lem:path_prefixes}
    Let $\pi$ be an $\ell$-length path in the NFA $\mathcal{A}$ with $\iota(w_{\pi}) = c$, and $\al(\r(w_{\pi})) = \mathcal{R}$.
    Now, we distinguish the cases whether $\mathcal{R}$ is empty:
    \begin{description}
    \item[$\mathcal{R}=\emptyset$:]  $\iota(w_{\pi}[1, \ell - 1]) = c - 1$ and $\al(\r(w_{\pi}[1, \ell - 1])) = \Sigma \setminus \{ w_{\pi}[\ell] \}$,
    \item[$\mathcal{R}\neq\emptyset$:] $\iota(w_{\pi}[1, \ell - 1]) = c$ and 
    either $\al(\r(w_{\pi}[1, \ell - 1])) = \al(r(w_{\pi})) \setminus \{w_{\pi}[\ell] \}$ or $\al(\r(w_{\pi}[1, \ell - 1])) = \al(\r(w_{\pi}))$. In this case we also have $\mathcal{R} = \al(r(w_{\pi}[1, \ell - 1])) \cup \{w_{\pi}[\ell] \}$.
    \end{description}
\end{corollary}

Now, we formally establish the dynamic programming approach to calculate the value of $T[q,\ell,c,\mathcal{R}]$ from the already computed cells of the table. For better readability, we use $\mathcal{P}[q,\ell,c,\mathcal{R}]$ as the set of all $\ell$-length paths from $q_0$ to $q$ where the associated word has $c$ arches and the alphabet of its rest is $\mathcal{R}$. Notice that we have $|\mathcal{P}[q,\ell,c,\mathcal{R}]|=T[q,\ell,c,\mathcal{R}]$.

\begin{lemma}
    \label{lem:number_of_paths}
    Let $\mathcal{A}$ be an NFA and assume that $T[q_0,0,0,\emptyset]$ is given.
    Notice that given $q\in Q$, the combination $(q',\tx)\in Q\times\Sigma$ only contributes to $T[q,\ell,c,\mathcal{R}]$ if we have $q'\in\Delta(q,\tx)$. Thus, we have for all $\ell\geq 1$    
    $$ 
    T[q, \ell, c, \mathcal{R}] =
        \sum\limits_{\overset{\tx \in \Sigma,}{q'\in\Delta(q,\tx)}} \begin{cases}
            0 & \mbox{if }\tx \notin \mathcal{R} \text{ and } \mathcal{R} \neq \emptyset,\\
            0 & \mbox{if }\mathcal{R} = \emptyset, c = 0,\\
            T[q', \ell - 1, c - 1, \Sigma \setminus \{\tx\}] & \mbox{if } \mathcal{R} = \emptyset, c > 0, \\
                T[q', \ell - 1, c, \mathcal{R} \setminus \{\tx\}] 
                + T[q', \ell - 1, c, \mathcal{R}]
                 & \mbox{if }\mathcal{R} \neq \emptyset, x \in \mathcal{R}.
        \end{cases}
    $$
\end{lemma}
\begin{proof}
	From Lemma \ref{lem:making_a_path_longer}, for every path $\pi \in \mathcal{P}[q, \ell - 1, c, \mathcal{R}]$, corresponding to the word $w_{\pi}$, and symbol $\tx$, there exists some $\ell$-length path $\pi'$ ending at some state $q' \in \delta(q,x)$ corresponding to the word $w_{\pi'} = w_{\pi} \tx$ such that:
	\begin{itemize}
		\item $\pi' \in \mathcal{P}[q', \ell, c + 1, \emptyset]$, if $\mathcal{R} = \Sigma \setminus \{ \tx \}$, or
		\item $\pi' \in \mathcal{P}[q', \ell, c, \mathcal{R} \cup \{ x \}]$ if $\mathcal{R} \neq \Sigma \setminus \{ \tx \}$.
	\end{itemize}
	
	In the other direction, from Lemma \ref{lem:path_prefixes}, every path $\pi' \in \mathcal{P}[q, \ell, c, \mathcal{R}]$ corresponding to the word $w_{\pi'}$ must contain an $(\ell - 1)$-length prefix corresponding to a word with either:
	\begin{itemize}
		\item $c - 1$ arches, and the set of symbols $\Sigma \setminus \{\tx\}$ in the rest of the word, if $\mathcal{R} = \emptyset$, or
		\item $c$ arches, and rest of the word containing the set of symbols $\mathcal{R}$ or $\mathcal{R} \setminus \{ \tx \}$, if $\mathcal{R} \neq \emptyset$.
	\end{itemize}
	Therefore, if $\mathcal{R} = \emptyset$, the size of $\mathcal{P}[q, \ell, c, \mathcal{R}]$ is equal to $\sum_{x \in \Sigma} \sum_{q' \in \Delta(q, \tx)} T[q', \ell - 1, c - 1, \Sigma \setminus \{ \tx \}]$.
	Similarly, if $\mathcal{R} \neq \emptyset$, the size of $\mathcal{P}[q, \ell, c, \mathcal{R}]$ is equal to $\sum_{x \in \mathcal{R}} \sum_{q' \in \Delta(q,\tx)} T[q', \ell - 1, c, \mathcal{R} \setminus \{x\}] + T[q', \ell - 1, c, \mathcal{R}]$.
	
	Note that if $\mathcal{R} \neq \emptyset$, and $\tx \notin \mathcal{R}$ for some $\tx\in\Sigma$, there is no path in $\mathcal{P}[q, \ell, c, \mathcal{R}]$ where the final transition is labelled $\tx$. 
	Similarly , $\mathcal{P}[q, \ell, 0, \emptyset] = \emptyset$ for any $\ell \geq 1$.
	Otherwise, one of the two above cases must apply, depending on the value of $\mathcal{R}$. This concludes the proof.
\end{proof}

Following Lemma \ref{lem:number_of_paths}, the table $T$ can be constructed via dynamic programming. 
As a base case, note that the only length $0$ path in this automaton starting at $q_0$ is the empty path, which must also end at state $q_0$ and contains $0$ arches and no symbols in the alphabet of the rest.
Therefore, $T[q, 0, c, \mathcal{R}]$ is set to $0$ for every $q \in Q, \mathcal{R} \subsetneq \Sigma$ and $c \in [0, k-1]$ other than $T[q_0, 0, 0, \emptyset]$, which is set to 1.
We now use $T$ to construct $U$.
%

\begin{corollary}
    \label{col:generally_computing_the_path_table}
    We have
    $U[q,\ell]= \sum_{\tx \in \Sigma,q'\in\Delta(q,\tx)} 
    		U[q', \ell - 1] + T[q', \ell - 1, k - 1, \Sigma \setminus \{\tx\}]$.
\end{corollary}
\begin{proof}
    The correctness of this construction follows from Lemma \ref{lem:number_of_paths}.
%
\end{proof}

Note that the total number of $k$-universal accepting paths of the automaton $\mathcal{A}$ is given by $\sum_{q \in F} U[q, m]$.
Using the tables $T$ and $U$, the total number of $k$-universal paths and words of length $m$ can be computed in 
$O^*(m^2 n^2 k 2^{\sigma})$.

\begin{theorem}
    \label{thm:counting_exactly_ell_k_subsequence_universal_words}
    The number of $k$-universal accepting paths of length $m$ of  an NFA $\mathcal{A}$, for $k\leq m/\sigma$, can be computed in 
    $O^*(m^2 n^2 k 2^{\sigma})$
    time. For  $k> m/\sigma$, this number is $0$.
\end{theorem}
\begin{proof}
	Note that, as explained before, all arithmetic operations require $O(m)$ time in our computational model (we operate with numbers of value at most $(n\sigma)^m$).
	    Using the table $U$, the number of $k$-universal words of length $m$ can be counted by computing the sum $\sum_{q \in F} U[q, m]$, which takes $O^*(m n)$ time.
	    To construct the table $U$, it is nessesary first to construct the table $T$.
	    Assuming that the value of $T[q', \ell - 1, c', \mathcal{R}']$ has been precomputed for every $q' \in Q, c' \in [1, k-1]$ and $\mathcal{R}' \subsetneq \Sigma$, the value of $T[q, \ell, c, \mathcal{R}]$ can be computed in $O^*(m n)$ time.
    As there are $n$ values of $q \in Q$, $2^{\sigma}$ values of $\mathcal{R} \subset \Sigma$, and $k$ values of $c \in [0, k-1]$,  the value of $T[q, \ell, c, \mathcal{R}]$ can be computed for every $q \in Q$, $c \in [0, k-1]$ , and $\mathcal{R} \subset \Sigma$ in $O^*(m n^2 k 2^{\sigma})$ time.
    As there are $m + 1$ values of $\ell \in [0, m]$, the total time of constructing $T$ is $O^*(m^2 n^2 k 2^{\sigma})$.
    
	Analogously, the value of $U[q, \ell]$ can be computed in $O^*(m n)$ time,
	assuming the values of $U[q', \ell - 1]$ and $T[q', \ell - 1, k - 1, \Sigma \setminus \{\tx\}]$ have be precomputed from every $q' \in Q$ and $\tx \in \Sigma$.
    Hence $U$ can be constructed from the table $T$ in $O^*(m^2 n^2)$ time. 
    By extension the number of $k$-universal paths of length $m$ accepted by $\mathcal{A}$ computed in $O^*(m^2 n^2 k 2^{\sigma})$ time.
\end{proof}

\begin{corollary}
    \label{col:counting_up_to_ell_k_subsequence_univeral_paths}
    The number of  $k$-universal accepting paths of length at most $m$ of an NFA $\mathcal{A}$, for $k\leq m/\sigma$, can be computed in 
    $O^*(m^2 n^2 k 2^{\sigma})$
    time. For  $k> m/\sigma$, this number is $0$.
\end{corollary}

\begin{proof}
	Observe that the number of $k$-universal paths of length at most $\ell$ is equal to $\sum_{q \in F} \sum_{\ell \in [1, m]} U[q, \ell]$.
	As this summation can be computed in 
	$O^*(m n)$ time once $U$ has been constructed, the total complexity follows from the cost of constructing tables $T$ and~$U$.
\end{proof}

Since in a DFA each word is associated with exactly one path, the following holds.

\begin{corollary}
    \label{col:counting_up_to_ell_k_subsequence_univeral_words}
    The number of $k$-universal words of length either exactly or at most $m$ accepted by a DFA $\mathcal{A}$, for $k\leq m/\sigma$, can be computed in 
    $O^*(m^2 n^2 k 2^{\sigma})$ 
    time. For  $k> m/\sigma$, this number is $0$.
\end{corollary}


We may further generalise this to the problem of finding the number of perfect $k$-universal words and paths by discarding any $k$-universal paths with a non-empty rest.

\begin{corollary}
	\label{col:counting_perfect_universal_paths}
	The number of perfect $k$-universal accepting paths of length either exactly or at most $m$ of an NFA $\mathcal{A}$, for $k\leq m/\sigma$, can be computed in 
	$O^*(m^2 n^2 k 2^{\sigma})$
	time.
\end{corollary}

\begin{corollary}
	\label{col:counting_perfect_universal_words}
	The number of perfect $k$-universal words of length either exactly or at most $m$  accepted by a DFA $\mathcal{A}$, for $k\leq m/\sigma$, can be computed in 
	$O^*(m^2 n^2 k 2^{\sigma})$
	time.
\end{corollary}

%

We now generalise these tools to the problem of counting the total number of $k$-universal paths accepted by a finite automaton $\mathcal{A}$.
The primary challenge of counting the total number of such paths comes from determining if there exists either a finite or an infinite number of $k$-universal paths accepted by $\mathcal{A}$.
By the pumping lemma \cite{HopcroftU79}, we have that an automaton $\mathcal{A}$ accepts an infinite number of $k$-universal words (paths) if and only if $\mathcal{A}$ accepts some $k$-universal word (path) of length at least $n + 1$. Clearly, if $k>n$ we have that $\mathcal{A}$ either accepts an infinite number of $k$-universal words (paths) or none (and this can be tested as in Lemma \ref{lem:ESUPFA}, in time that does not depend on $k$). Therefore, using the upper bound on the maximum length of the shortest $k$-universal word accepted by the automaton $\mathcal{A}$ from Lemma \ref{lem:ksu_min_length}, 
combined with Corollary \ref{col:counting_up_to_ell_k_subsequence_univeral_words}, we now count the total number of $k$-subsequence universal words accepted by $\mathcal{A}$.\looseness=-1

\begin{theorem}
\label{thm:counting_everything}
    The total number of $k$-universal words (resp., paths) accepted by a DFA (resp., NFA) $\mathcal{A}$ can be determined in 
    $O^*(n^4 k^2 2^{\sigma})$ time (resp.,  $O^*(n^4 k^3 2^{\sigma})$ time),
   for $k\leq n$. For $k> n$, this number is either $0$ or $\infty$, and can be determined in
     $O^*(n^3 2^\sigma)$ time. 
\end{theorem}
\begin{proof}
We only show this for NFAs, as the argument for DFAs is similar (and uses the previous results corresponding to this class of automata). 

    Note first that if an automaton accepts some $k$-universal paths $w$ of length at least $n + 1$, then it must accept an infinite number of such words, as the path induced by $w$ in $\mathcal{A}$ must visit some state twice and thus contain a cycle.
    Therefore, if $\mathcal{A}$ accepts only a finite number of $k$-universal paths, the total number of paths accepted by $\mathcal{A}$ can be computed in 
    $O^*(n^4 k 2^{\sigma})$ time via Theorem  \ref{thm:counting_exactly_ell_k_subsequence_universal_words} and Corollary \ref{col:counting_up_to_ell_k_subsequence_univeral_paths}.

    %
    Following the same arguments as in \Cref{lem:ksu_min_length}, a $k$-universal path of length of at least $n + 1$ exists if and only if there exists some $k$-universal word of length between $n + 1$ and $k n \sigma$, hence it is sufficient to check if $\mathcal{A}$ accepts some word of length between $n + 1$ and $k n \sigma$, which can be achieved in $O^*(n^4 k^3 2^{\sigma})$ time. 
    If no $k$-universal word accepted by $\mathcal{A}$ of length at least $n + 1$ exists, then total number of $k$-universal words accepted by $\mathcal{A}$ is determined by counting the number of $k$-universal words accepted by $\mathcal{A}$ of length at most $n$.
\end{proof}

We now extend our results of counting to the problem of ranking $k$-universal words.
Note, that these results can be generalised to ranking the $k$-universal accepted paths, but one needs to define an ordering on the transitions from each state in the automaton. To keep the presentation simple, we will therefore only discuss here about DFAs and words.
The main idea is to count the number of $k$-universal words with a prefix strictly smaller than the prefix of $w$ of the same length; again, each arithmetic operation takes $O(m)$ time in our computational model.
This section is laid out as follows.
First, we show how to compute the rank of $w$ efficiently within the set $\Univ_{\mathcal{A}_{\leq m}, k}$ in 
$O^*(m^2 n^2 k 2^{\sigma})$ time.
Secondly, we show that the rank of $w$ can be computed within the set $\Univ_{L(\mathcal{A}), k}$ in 
$O^*(n^4 k^3 2^{\sigma})$ time.
As noted above, when discussing counting, these problems only make sense for $k\leq m/\sigma$. 
This follows from the same arguments used to count the total number of $k$-universal words accepted by $\mathcal{A}$ laid out in Theorem \ref{thm:counting_everything}. \looseness=-1





The primary tool used in this section is a generalisation of the path table: the \emph{fixed prefix path table of length $m$} is an $n \times (m + 1)\times k \times 2^{\sigma}$ sized table, defined for a set of prefixes $\mathcal{PR}$ and denoted $T(\mathcal{PR})$.
Informally, the table $T(\mathcal{PR})$ is used to count the number of paths with some prefix from the given set 
$\mathcal{PR}$. Thus, $T(\mathcal{PR})[q, \ell, c, \mathcal{R}]$ stores the number of words $w\in\Sigma^{\ell}$
associated to a path from $q_0$ to $q$, with $\iota(w)=c$, $\letters(r(w))=\mathcal{R}$, and there exists a $p\in\mathcal{PR}$ such that $p=w[1,|p|]$.
The table $U(\mathcal{PR})$ is defined analogously to the table $U$ but again with the additional condition as in $T(\mathcal{PR})$.
These tables can be constructed directly using the same techniques introduced for $T$ and $U$, by initially setting $T(\mathcal{PR})[q_0, 0, 0, \emptyset]$ to 0 and $T(\mathcal{PR})[\delta(q_0, p), \lvert p \rvert, \iota(p), \letters(\r(p))]$ to $1$ for every $p \in \mathcal{PR}$. 
Similarly, in the special case where there exists some $p \in \mathcal{PR}$ such that $\iota(p) \geq k$, then the value of $U[\delta(q_0, p), \lvert p \rvert ]$ is set to 1.
We assume, without loss of generality, that no prefix in $\mathcal{PR}$ is also the prefix of some other word $p' \in \mathcal{PR}$.
The remaining entries are computed as before.
In the following results, we use $\mathcal{PR}(w) = \{w[1, i] \tx \mid i \in [0, \lvert w \rvert ], \tx \in [1, w[i + 1] - 1] \}$.
Note, that the following results hold for both counting words accepted by deterministic automata, and accepting paths in non-deterministic automata.

\begin{corollary}
    \label{col:constructing_the_fixed_perfix_table}
    $T(\mathcal{PR})$, $U(\mathcal{PR})$ are constructible for $m$-length paths in 
    $O^*(m^2 n^2 k 2^{\sigma})$ time.
\end{corollary}


\begin{theorem}
    \label{thm:ranking_fixed_length}
    The rank of $w \in \Univ_{\mathcal{A}_m, k}$ can be determined
    in $O^*(m^2 n^2 k 2^{\sigma})$ time.
\end{theorem}
\begin{proof}
    Note that a word $v$ is smaller than $w$ (w.r.t. the lexicographical ordering) if and only if $v$ is a prefix of $w$ or they share a common prefix $u$ and $v[|u|+1]<w[|u|+1]$.
    Therefore, the number of $m$-length $k$-universal words starting with some prefix in $\mathcal{PR}(w)$ is given by either $\sum_{q \in F} U(\mathcal{PR}(w))[q, m]$, if $w$ has length at most $m$, or $1+\sum_{q \in F} U(\mathcal{PR}(w))[q, m]$ if the \nth{$w[1,m]$} state of the path associated with $w$ is an accepting state, $\lvert w \rvert > m$, and $\iota(w[1, m]) = k$.
    As $T(\mathcal{PR}(w))$ can be computed in 
    $O^*(m^2 n^2 k 2^{\sigma})$ time, 
    and the above summation completed in 
    $O^*(m n)$ time, the total time complexity of finding the $m$-length rank of $w$ is $O^*(m^2 n^2 k 2^{\sigma})$.
\end{proof}

\begin{corollary}
    \label{col:ranking_up_to_k_length}
    The rank of $w\in \Univ_{\mathcal{A}_{\leq m}, k}$ can be determined 
    in $O^*(m^2 n^2 k  2^{\sigma})$ time.
\end{corollary}
\begin{proof}
    Using the table $T(\mathcal{PR}(w))$ as above, the number of words $k$-universal words of length at most $m$ smaller than $w$ is given by 
    \[
    \sum_{i \in [1, m]} \sum_{q \in F} U(\mathcal{PR}(w))[q, i] + \sum_{i \in [1, m]} \begin{cases}
        1, & q_{w[1, i]} \in F,\\
        0, & q_{w[1, i]} \notin F.
    \end{cases}
    \]
    As the table can be constructed in 
    $O^*(m^2 n^2 k 2^{\sigma})$ time,  and the summation requires at most 
    $O^*(m^2 n)$ time, the total complexity of finding the at-most-$m$-length rank of a word $w$ in $\Univ_{\mathcal{A}_m, k}$
    is $O^*(m^2 n^2 k 2^{\sigma})$.
\end{proof}

\begin{corollary}
    \label{col:ranking_everything}
    The rank of $w \in \Univ_{L(\mathcal{A}), k}$ can be determined in $O^*(n^4 k^3 2^{\sigma})$ time.
\end{corollary}
\begin{proof}
    Following the same arguments as given in Theorem \ref{thm:counting_everything}, note that there is an infinite number of words smaller than $w$ if and only if there exists some word of length at least $n + 1$ with a prefix in $\mathcal{PR}$.
    The existence of such a word can be determined from the tables $T(\mathcal{PR}(w))$ and $U(\mathcal{PR}(w))$ for paths of length at most $k m \sigma$ in $O^*(k^2 n^3)$ time. 
    As the tables $T(\mathcal{PR}(w))$ and $U(\mathcal{PR}(w))$ can be constructed for paths of length at most $k n \sigma$ in 
    $O^*(n^4 k^3 2^{\sigma})$ time, the total rank of $w$ within $\Univ_{L(\mathcal{A}), k}$ can be computed in $O^*(n^4 k^2 )$ time. 
\end{proof}

\section{Conclusions}\label{conc}
This paper proposed a series of novel algorithmic results and insights regarding the analysis of the sets which can be expressed as the intersection of regular languages and the language of $k$-universal words over some alphabet.
We have introduced two natural notions of $k$-universality in regular languages, namely existence $k$-universal languages and universal $k$-universal languages, and have proposed algorithms for testing whether a regular language is in one of these two classes. While we have a good understanding of the problem of deciding whether a language is defined by an existence $k$-universal automaton, the exact complexity of the problem of deciding whether a language is universal $k$-universal remains open. As well as the introduction of these notions, and the study of some decisions problems related to them, we have provided a toolbox for counting and ranking $k$-universal paths (respectively, words) accepted by a given NFA (respectively, DFA).

We note that using a divide and conquer approach to count paths (with a certain amount of arches, at most $m$) of length $2^\ell$ by combining paths of length $2^{\ell-1}$ (with less arches), one factor $m$ can be reduced to $\log m$ for counting and ranking words of (or of at most) given length $m$, at the cost of additional complexity in terms of $n, k$ and $\sigma$ (as, in that case, one would have to allow the existence of prefixes of such paths which are not part of arches, as well as consider the fact that these paths connect arbitrary pairs of states, and may have different counts of arches). This has been left out of the current version to provide a clearer understanding of our main results, and avoid over-complicating the presentation of the paper.\looseness=-1

Duncan Adamson's work was funded by the Leverhulme Trust via the Leverhulme Research Centre for Functional Material Design. Tore Koß’ work was supported by the DFG project number 389613931. Florin Manea’s work was supported by the DFG Heisenberg-project number 466789228.\looseness=-1

\bibliography{refs}

\begin{thebibliography}{10}

\bibitem{adamson2022ranking}
D.~Adamson.
\newblock Ranking binary unlabelled necklaces in polynomial time.
\newblock In {\em DCFS}, pages 15--29. Springer, 2022.

\bibitem{adamson2023words}
D.~Adamson.
\newblock Ranking and unranking $k$-subsequence universal words.
\newblock In Anna Frid and Robert Merca{\c{s}}, editors, {\em WORDS}, pages
  47--59. Springer Nature Switzerland, 2023.

\bibitem{Adamson2021}
D.~Adamson, A.~Deligkas, V.~V. Gusev, and I.~Potapov.
\newblock Ranking bracelets in polynomial time.
\newblock {\em CPM}, pages 4--17, 2021.

\bibitem{Goettingen2023words}
D.~Adamson, M.~Kosche, T.~Ko{\ss}, F.~Manea, and S.~Siemer.
\newblock Longest common subsequence with gap constraints.
\newblock In Anna Frid and Robert Merca{\c{s}}, editors, {\em WORDS}, pages
  60--76, 2023.

\bibitem{artikis2017complex}
A.~Artikis, A.~Margara, M.~Ugarte, S.~Vansummeren, and M.~Weidlich.
\newblock Complex event recognition languages: Tutorial.
\newblock In {\em DEBS}, pages 7--10, 2017.

\bibitem{bachmeier2015finite}
Georg Bachmeier, Michael Luttenberger, and Maximilian Schlund.
\newblock Finite automata for the sub-and superword closure of cfls:
  Descriptional and computational complexity.
\newblock In {\em International Conference on Language and Automata Theory and
  Applications}, pages 473--485. Springer, 2015.

\bibitem{barker2020scattered}
L.~Barker, P.~Fleischmann, K.~Harwardt, F.~Manea, and D.~Nowotka.
\newblock Scattered factor-universality of words.
\newblock In {\em DLT}, pages 14--28. Springer, 2020.

\bibitem{ChenKMS17}
H.~Z.~Q. Chen, S.~Kitaev, T.~M{\"u}tze, and B.~Y. Sun.
\newblock On universal partial words.
\newblock {\em Electronic Notes in Discrete Mathematics}, 61:231--237, 2017.

\bibitem{crochemore}
M.~Crochemore, C.~Hancart, and T.~Lecroq.
\newblock {\em Algorithms on strings}.
\newblock Cambridge University Press, 2007.

\bibitem{day2021edit}
J.D. Day, P.~Fleischmann, M.~Kosche, T.~Ko{\ss}, F.~Manea, and S.~Siemer.
\newblock The edit distance to k-subsequence universality.
\newblock In {\em {STACS}}, volume 187, pages 25:1--25:19, 2021.

\bibitem{Bruijn46}
N.~G. de~Bruijn.
\newblock A combinatorial problem.
\newblock {\em Koninklijke Nederlandse Akademie v. Wetenschappen}, 49:758--764,
  1946.

\bibitem{fleischer2018testing}
L.~Fleischer and M.~Kufleitner.
\newblock Testing simon's congruence.
\newblock In {\em MFCS}. Schloss Dagstuhl-Leibniz-Zentrum fuer Informatik,
  2018.

\bibitem{fleischmann2021scattered}
P.~Fleischmann, S.B. Germann, and D.~Nowotka.
\newblock Scattered factor universality--the power of the remainder.
\newblock {\em preprint arXiv:2104.09063 (published at RuFiDim)}, 2021.

\bibitem{fleischmann2022nearly}
P.~Fleischmann, L.~Haschke, A.~Huch, A.~Mayrock, and D.~Nowotka.
\newblock Nearly k-universal words-investigating a part of simon’s
  congruence.
\newblock In {\em DCFS}, pages 57--71, 2022.

\bibitem{fleischmann2023alphabetafactorization}
P.~Fleischmann, J.~Höfer, A.~Huch, and D.~Nowotka.
\newblock $\alpha$-$\beta$-factorization and the binary case of simon's
  congruence, 2023.
\newblock \href {https://arxiv.org/abs/2306.14192} {\path{arXiv:2306.14192}}.

\bibitem{FominKTV08}
F.~V. Fomin, D.~Kratsch, I.~Todinca, and Y.~Villanger.
\newblock Exact algorithms for treewidth and minimum fill-in.
\newblock {\em {SIAM} J. Comput.}, 38(3):1058--1079, 2008.
\newblock \href {https://doi.org/10.1137/050643350}
  {\path{doi:10.1137/050643350}}.

\bibitem{Fredricksen1978}
H.~Fredricksen and J.~Maiorana.
\newblock {Necklaces of beads in k colors and k-ary de Bruijn sequences}.
\newblock {\em Discrete Mathematics}, 23(3):207--210, 1978.

\bibitem{FrochauxK23}
A.~Frochaux and S.~Kleest{-}Mei{\ss}ner.
\newblock Puzzling over subsequence-query extensions: Disjunction and
  generalised gaps.
\newblock In {\em {AMW} 2023}, volume 3409 of {\em {CEUR} Workshop
  Proceedings}. CEUR-WS.org, 2023.

\bibitem{gawrychowski2021simons}
P.~Gawrychowski, M.~Kosche, T.~Ko{\ss}, F.~Manea, and S.~Siemer.
\newblock {Efficiently Testing Simon’s Congruence}.
\newblock In {\em STACS}, volume 187, pages 34:1--34:18, 2021.

\bibitem{GawrychowskiRSS17}
P.~Gawrychowski, M.~Lange, N.~Rampersad, J.~O. Shallit, and M.~Szykula.
\newblock Existential length universality.
\newblock In {\em Proc. {STACS} 2020}, volume 154 of {\em LIPIcs}, pages
  16:1--16:14, 2020.

\bibitem{gilbert1961symmetry}
E.~N. Gilbert and J.~Riordan.
\newblock Symmetry types of periodic sequences.
\newblock {\em Illinois Journal of Mathematics}, 5(4):657--665, 1961.

\bibitem{GoecknerGHKKKS18}
B.~Goeckner, C.~Groothuis, C.~Hettle, B.~Kell, P.~Kirkpatrick, R.~Kirsch, and
  R.~W. Solava.
\newblock Universal partial words over non-binary alphabets.
\newblock {\em Theor. Comput. Sci}, 713:56--65, 2018.

\bibitem{halfon2017decidability}
S.~Halfon, P.~Schnoebelen, and G.~Zetzsche.
\newblock Decidability, complexity, and expressiveness of first-order logic
  over the subword ordering.
\newblock In {\em LICS}, pages 1--12. IEEE, 2017.

\bibitem{han2020novel}
R.~Han, S.~Wang, and X.~Gao.
\newblock Novel algorithms for efficient subsequence searching and mapping in
  nanopore raw signals towards targeted sequencing.
\newblock {\em Bioinformatics}, 36(5):1333--1343, 2020.

\bibitem{hebrard1991algorithm}
J.-J. Hebrard.
\newblock An algorithm for distinguishing efficiently bit-strings by their
  subsequences.
\newblock {\em Theoretical Computer Science}, 82(1):35--49, 1991.

\bibitem{HolzerK11}
M.~Holzer and M.~Kutrib.
\newblock Descriptional and computational complexity of finite automata - {A}
  survey.
\newblock {\em Inf. Comput.}, 209(3):456--470, 2011.

\bibitem{HopcroftU79}
J.~E. Hopcroft and J.~D. Ullman.
\newblock {\em Introduction to Automata Theory, Languages and Computation}.
\newblock Addison-Wesley, 1979.

\bibitem{KarandikarKS15}
P.~Karandikar, M.~Kufleitner, and P.~Schnoebelen.
\newblock On the index of {S}imon's congruence for piecewise testability.
\newblock {\em Inf. Process. Lett.}, 115(4):515--519, 2015.

\bibitem{karandikar2016height}
P.~Karandikar and P.~Schnoebelen.
\newblock The height of piecewise-testable languages with applications in
  logical complexity.
\newblock In {\em CSL}, 2016.

\bibitem{karandikar2016state}
Prateek Karandikar, Matthias Niewerth, and Ph~Schnoebelen.
\newblock On the state complexity of closures and interiors of regular
  languages with subwords and superwords.
\newblock {\em Theoretical Computer Science}, 610:91--107, 2016.

\bibitem{KimHKS22}
S.~Kim, Y.~Han, S.~Ko, and K.~Salomaa.
\newblock On simon's congruence closure of a string.
\newblock In {\em {DCFS} 2022, Proceedings}, volume 13439 of {\em Lecture Notes
  in Computer Science}, pages 127--141. Springer, 2022.

\bibitem{KimHKS23}
S.~Kim, Y.~Han, S.~Ko, and K.~Salomaa.
\newblock On the simon's congruence neighborhood of languages.
\newblock In {\em {DLT} 2023, Proceedings}, volume 13911 of {\em Lecture Notes
  in Computer Science}, pages 168--181. Springer, 2023.

\bibitem{KimKH22}
S.~Kim, S.~Ko, and Y.~Han.
\newblock Simon's congruence pattern matching.
\newblock In {\em {ISAAC} 2022, Proceedings}, volume 248 of {\em LIPIcs}, pages
  60:1--60:17. Schloss Dagstuhl - Leibniz-Zentrum f{\"{u}}r Informatik, 2022.

\bibitem{Kleest-Meissner22}
S.~Kleest{-}Mei{\ss}ner, R.~Sattler, M.~L. Schmid, N.~Schweikardt, and
  M.~Weidlich.
\newblock Discovering event queries from traces: Laying foundations for
  subsequence-queries with wildcards and gap-size constraints.
\newblock In {\em {ICDT} 2022, Proceedings}, volume 220 of {\em LIPIcs}, pages
  18:1--18:21, 2022.

\bibitem{Kleest-Meissner23}
S.~Kleest{-}Mei{\ss}ner, R.~Sattler, M.~L. Schmid, N.~Schweikardt, and
  M.~Weidlich.
\newblock Discovering multi-dimensional subsequence queries from traces - from
  theory to practice.
\newblock In {\em {BTW} 2023, Proceedings}, volume {P-331} of {\em {LNI}},
  pages 511--533, 2023.

\bibitem{Kociumaka2014}
T.~Kociumaka, J.~Radoszewski, and W.~Rytter.
\newblock {Computing $k$-th Lyndon word and decoding lexicographically minimal
  de Bruijn sequence}.
\newblock In {\em CPM}, pages 202--211. Springer, 2014.

\bibitem{kosche2021absent}
M.~Kosche, T.~Ko{\ss}, F.~Manea, and S.~Siemer.
\newblock Absent subsequences in words.
\newblock In {\em RP}, pages 115--131. Springer, 2021.

\bibitem{Kosche2022SubsequenceSurvey}
M.~Kosche, T.~Ko{\ss}, F.~Manea, and S.~Siemer.
\newblock Combinatorial algorithms for subsequence matching: A survey.
\newblock In Henning Bordihn, G\'eza Horv\'ath, and Gy\"orgy Vaszil, editors,
  {\em NCMA}, 2022.

\bibitem{KrotzschMT17}
M.~Kr{\"{o}}tzsch, T.~Masopust, and M.~Thomazo.
\newblock Complexity of universality and related problems for partially ordered
  {NFA}s.
\newblock {\em Inf. Comput.}, 255:177--192, 2017.

\bibitem{lothaire}
M.~Lothaire.
\newblock {\em Combinatorics on Words}.
\newblock Cambridge Mathematical Library. Cambridge University Press, 1997.

\bibitem{martin1934}
M.~H. Martin.
\newblock A problem in arrangements.
\newblock {\em Bull. Amer. Math. Soc.}, 40(12):859--864, 12 1934.

\bibitem{mateescu2004subword}
A.~Mateescu, A.~Salomaa, and S.~Yu.
\newblock Subword histories and parikh matrices.
\newblock {\em Journal of Computer and System Sciences}, 68(1):1--21, 2004.

\bibitem{Rampersad:2012}
N.~Rampersad, J.~Shallit, and Z.~Xu.
\newblock The computational complexity of universality problems for prefixes,
  suffixes, factors, and subwords of regular languages.
\newblock {\em Fundam. Inf.}, 116(1-4):223--236, January 2012.

\bibitem{savage1997survey}
C.~Savage.
\newblock A survey of combinatorial gray codes.
\newblock {\em SIAM review}, 39(4):605--629, 1997.

\bibitem{Sawada2017}
J.~Sawada and A.~Williams.
\newblock Practical algorithms to rank necklaces, lyndon words, and de bruijn
  sequences.
\newblock {\em Journal of Discrete Algorithms}, 43:95--110, 2017.

\bibitem{schnoebelen2019height}
P.~Schnoebelen and P.~Karandikar.
\newblock The height of piecewise-testable languages and the complexity of the
  logic of subwords.
\newblock {\em Logical Methods in Computer Science}, 15, 2019.

\bibitem{SchnoebelenV23}
P.~Schnoebelen and J.~Veron.
\newblock On arch factorization and subword universality for words and
  compressed words.
\newblock In {\em {WORDS} 2023, Proceedings}, volume 13899 of {\em Lecture
  Notes in Computer Science}, pages 274--287. Springer, 2023.

\bibitem{shaw1978software}
A.~C. Shaw.
\newblock Software descriptions with flow expressions.
\newblock {\em IEEE Transactions on Software Engineering}, 3:242--254, 1978.

\bibitem{shikder2019openmp}
R.~Shikder, P.~Thulasiraman, P.~Irani, and P.~Hu.
\newblock An openmp-based tool for finding longest common subsequence in
  bioinformatics.
\newblock {\em BMC research notes}, 12:1--6, 2019.

\bibitem{Simon72}
I.~Simon.
\newblock Piecewise testable events.
\newblock In {\em Autom.\ Theor.\ Form.\ Lang., 2nd GI Conf.}, volume~33 of
  {\em LNCS}, pages 214--222. Springer, 1975.

\bibitem{simon2003words}
I.~Simon.
\newblock Words distinguished by their subwords.
\newblock {\em WORDS}, 27:6--13, 2003.

\bibitem{tronicek2003common}
Z.~Troni{\^c}ek.
\newblock Common subsequence automaton.
\newblock In {\em CIAA}, pages 270--275, 2003.

\bibitem{zetzsche2016complexity}
G.~Zetzsche.
\newblock The complexity of downward closure comparisons.
\newblock In {\em ICALP}, volume~55, pages 123:1--123:14, 2016.

\end{thebibliography}

\newpage

\end{document}